\documentclass{article} 
\title{\textbf{Average-radius list-recoverability \\ of random linear codes}\\ \vspace{.3cm}
 \Large It really ties the room together}
\author{Atri Rudra\thanks{\texttt{atri@buffalo.edu}.  University at Buffalo, SUNY. Research supported in part by NSF grant\# CCF-1319402.} \and Mary Wootters\thanks{\texttt{marykw@stanford.edu}.  Stanford University.}}

\usepackage[margin=1.125in]{geometry}
\usepackage{xcolor}
\definecolor{DarkGreen}{rgb}{0.1,0.5,0.1}
\definecolor{DarkRed}{rgb}{0.5,0.1,0.1}
\definecolor{DarkBlue}{rgb}{0.1,0.1,0.5}

\usepackage[small]{caption}
\usepackage[pdftex]{hyperref}
\hypersetup{
    unicode=false,          
    pdftoolbar=true,        
    pdfmenubar=true,        
    pdffitwindow=false,      
    pdfnewwindow=true,      
    colorlinks=true,       
    linkcolor=DarkBlue,          
    citecolor=DarkGreen,        
    filecolor=DarkGreen,      
    urlcolor=DarkBlue,          
    pdftitle={},
    pdfauthor={},    
}
\usepackage{amssymb,amsmath,amsthm,amsfonts,enumitem}
\usepackage{fullpage,nicefrac,comment}
\usepackage{tikz}
\usetikzlibrary{arrows,decorations.pathmorphing,decorations.shapes,snakes,patterns,positioning}
\usepackage[ruled,linesnumbered]{algorithm2e}

\newcommand{\cC}{\ensuremath{\mathcal{C}}}
\newcommand{\cE}{\ensuremath{\mathcal{E}}}

\newcommand{\F}{{\mathbb F}}

\newcommand{\bX}{\mathbf{X}}

\newcommand{\PR}[1]{{\mathbb{P}}\left\{ #1\right\}}

\newcommand{\EE}{\mathbb{E}}

\newcommand{\inabs}[1]{\left|#1\right|}
\newcommand{\ip}[2]{\ensuremath{\left\langle #1,#2\right\rangle}}

\newcommand{\inset}[1]{\left\{#1\right\}}

\newcommand{\inparen}[1]{\left(#1\right)}

\newcommand{\suchthat}{\,:\,}

\newcommand{\pl}{\mathrm{pl}}
\newcommand{\cols}{\mathrm{cols}}

\newcommand{\argmax}{\mathrm{argmax}}

\newcommand{\polylog}{\mathrm{polylog}}
\newcommand{\poly}{\mathrm{poly}}

\newcommand{\vol}{\ensuremath{\operatorname{Vol}}}

\newcommand{\spn}{\ensuremath{\operatorname{span}}}
\newcommand{\dist}{\ensuremath{\operatorname{dist}}}

\newcommand{\ind}[1]{\ensuremath{\mathbf{1}_{#1}}}

\newcommand{\eps}{\varepsilon} 
\renewcommand{\epsilon}{\varepsilon}

\newcommand{\ar}[1]{{\footnotesize  [\textbf{\textcolor{blue}{#1}} \textcolor{blue!60!black}{--Atri}]\normalsize}}
\newcommand{\TODO}[1]{\textcolor{red}{\textbf{TODO: #1}}}

\newtheorem{theorem}{Theorem} [section]
\newtheorem{lemma}[theorem]{Lemma} 
\newtheorem{definition}[theorem]{Definition}

\newtheorem{corollary}[theorem]{Corollary} 

\newtheorem{remark}{Remark}
\newtheorem{claim}[theorem]{Claim}
\newtheorem{subclaim}[theorem]{Sub-Claim}
\newtheorem{proposition}[theorem]{Proposition}

\newtheorem{question}[theorem]{Question}

\newcommand{\nmrk}[1]{\textcolor{red}{{#1}}}
\renewcommand{\nmrk}{}
\newcommand{\arm}[1]{\textcolor{red}{{#1}}}
\renewcommand{\arm}{}

\newcommand{\mrk}[1]{#1}
\newcommand{\tp}{\mathrm{top}}
\newcommand{\argtp}{\mathrm{argtop}}
\newcommand{\ml}{\mrk{,\ell}}
\newcommand{\dmax}{d_{j_{\max}}}

\newcommand{\barmu}{\nmrk{\overline{\mu}}}
\newcommand{\centermu}{\nmrk{\frac{\ell}{q}}}
\newcommand{\onebyq}{\nmrk{\frac{1}{q}}}

\begin{document}
\maketitle

\setcounter{page}{0}
\thispagestyle{empty}

\begin{abstract}
	We analyze the list-decodability, and related notions, of random linear codes.  This has been studied extensively before: there are many different parameter regimes and many different variants.  Previous works have used complementary styles of arguments---which each work in their own parameter regimes but not in others---and moreover have left some gaps in our understanding of the list-decodability of random linear codes.  In particular, 
none of 
these arguments work well for list-\emph{recovery}, a generalization of list-decoding that has been useful in a variety of settings.

	In this work, we present a new approach, which works across parameter regimes and further generalizes to list-recovery.
	This argument unifies the landscape of this problem, and can 
	establish the following results about list-decoding and list-recovery:
\begin{itemize}
	\item \textbf{Better results for list-decoding and list-recovery over large fields.}  We show that random linear codes over large fields are list-recoverable and list-decodable up to near-optimal rates (within a multiplicative factor of $0.99$), with list sizes that depend quasi-polynomially on the gap-to-capacity.  Previous results for list-decoding had extraneous log factors and could not even guarantee a constant rate for super-constant field sizes; moreover they did not extend well to list-recovery.
	\item \textbf{Quasipolynomial  list sizes for high-rate list-recovery of random linear codes.}  While we know of several constructions of high-rate list-recoverable codes with small list sizes, to the best of our knowledge, none of them are linear.   Linearity is desirable, as such codes are used as building blocks in other coding-theoretic constructions.  Previous results could only guarantee list sizes that were exponential in the parameters of interest; our approach can obtain quasipolynomial list sizes.
	\item \textbf{Optimal-rate average-radius list-decoding over constant alphabets.} 
Using the same proof as for our other results, for a large range of parameters, we can match the optimal rate guarantees (albeit with a quasipolynomial dependence in the gap-to-capacity in the list size instead of polynomial) implied by the results of Guruswami, H{\aa}stad, and Kopparty for average-radius list-decoding.  This does not yield any new results, but it illustrates the generality of our approach.
\end{itemize}

\end{abstract}
\renewcommand{\arraystretch}{1.4}
	\clearpage
\section{Introduction}\label{sec:intro}
The \em list-decodability \em of random linear codes, and related notions, has been studied for decades in a variety of parameter regimes~\cite{ZP, elias91, GHK11, cgv2012, woot2013, RW14}.  

One reason that this has been studied so much is motivations throughout coding, complexity, and cryptography.  For example, in coding theory, list-decodability and list-recoverability of random linear codes is useful as a building block for other coding-theoretic constructions~\cite{GI01,HW15}, and ideas from these explorations have been useful in answering other coding-theoretic questions, for example the list-decodability of Reed-Solomon codes~\cite{RW14,RW15}.  In complexity theory, list-decodability and related notions are tied to notions in pseudorandomness like extractors, expanders, and pseudorandom generators~\cite{vadhan2011}; in cryptography the task of algorithmically list-decoding random linear codes is related to learning with errors (LWE) and learning parities with noise (LPN) and is assumed to be a hard problem.

Another reason that this line of questioning has attracted so much study is that, although couched in the language of coding theory, these questions are getting at a much deeper mathematical question about the geometry of random subspaces:
\begin{question}\label{q:bigq}
To what extent does a random subspace over a finite field behave (combinatorially speaking) like a completely random set of the same size?
\end{question}
Question~\ref{q:bigq} can be seen as a fundamental question about derandomization.  Completely random sets are, generally speaking, combinatorially very well-behaved, and seeing this usually amounts to using independence between the points to obtain concentration and taking a union bound.  The points of a random subspace are pairwise independent but are not even three-wise independent, and such approaches generally fail.

To state Question~\ref{q:bigq} more precisely, let $\cC \subset \F_q^n$ be a random $k$-dimensional subspace in $\F_q^n$.  In coding-theoretic language, $\cC$ is a \em random linear code \em of \em rate \em $R = k/n$.  

The first question one might ask is about pairwise distances: what is the minimum distance between any two points in a random subspace of dimension $k$ in $\F_q^n$?  How does this compare with the minimum distance between any two points of a random set of size $q^k$?  In this case the answer is well-understood.  A random subspace (aka, a random linear code) lies on the \em Gilbert-Varshamov bound, \em and the pairwise distances are similar to that of a (slightly modified) completely random set.  

The next question one might ask is about \em sets \em of points: what is the largest radius of any  set of points of size $L$ in a random subspace?  Phrased differently, how many points of a random subspace lie in any Hamming ball of a particular radius?  Does this quantity behave like it does for a completely random set of size $q^k$?  In coding-theoretic language, this is the \em list-decoding \em question.  
There are many related notions: how many points lie in any combinatorial cube?  This is the question of \em list-recovery. \em  How clustered are any set of $L$ points?  This is the question of \em average-radius list-decoding. \em

  Our primary motivation for this work is to further our understanding of Question~\ref{q:bigq}, which underlies all of these coding theoretic questions.

\subsection{Basic definitions and context}
We discuss related work in more detail when we present our results below in Section~\ref{sec:results}, but first we briefly survey what is known about Question~\ref{q:bigq}.  The most basic notion (after pairwise distance) that we might study is \em list decoding. \em
\begin{definition}
A code $\cC \subseteq \F^n$ is \em $(\rho, L)$-list-decodable, if for all sets $\Omega \subset \cC$ of size $|\Omega| \geq L$, we have, for all $z \in \F^n$,
\[ \max_{c \in \Omega} \dist(c, z) \geq \rho. \]
\end{definition}
Above, $\dist(x,y) := \frac{1}{n} \sum_{i=1}^n \ind{x_i \neq y_i}$ is relative Hamming distance.
Notice that an equivalent definition is that $\cC$ is $(\rho, L)$ list-decodable if for all $z$, the \em list size \em $\inabs{ \inset{ c \in \cC \suchthat \dist(c, z) < \rho } }$ is strictly less than $L$.

List-decodability was introduced by Elias and Wozencraft~\cite{elias,wozencraft} in the 1950's; one motivation is that it turns out that there are $(\rho, L)$-list-decodable codes for very large $\rho$ (approaching $1$) and for very small $L$ (constant).  This is somewhat surprising: it means that moderately reliable communication is possible even when an adversary is allowed to corrupt nearly $100\%$ of the transmitted symbols.  More precisely, the \em list-decoding capacity theorem \em says that there are codes of rate $R = k/n = R^* - \delta$, where \[R^* = 1 - H_q(\rho),\] which are $(\rho, L)$ list-decodable, where $L = \poly(1/\delta)$ does not depend on $n$, and $H_q(x)$ is the $q$-ary entropy (defined in~\eqref{eq:entropydef}).  Moreover, this rate $R^*$ is optimal, in the sense that $L$ must be exponentially large in $n$ for any code of rate  significantly larger than $R^*$.
The existence proof is in fact a random coding argument:  with high probability, a completely random set $\cC \subset \F_q^n$ of size $q^k$, for $k/n = 1 - H_q(\rho) - \delta$ will be $(\rho, O(1/\delta))$-list-decodable.

Returning to Question~\ref{q:bigq}, it is natural to ask whether a random subspace $\cC \subset \F_q^n$ of dimension $k$ (and hence size $q^k$) does just as well as a completely random subset.  On the one hand, one might expect this: after all, such a code $\cC$ is easily seen to have good distance, why not good list-decodability?  On the other hand, the argument for completely random codes crucially uses the independence of the elements of $\cC$, which a random subspace does not have.  In 1981, Zyablov and Pinsker gave an argument which shows that $k/n$ can be as large as $R^* - \delta$, but with a list-size exponentially large in $\delta$.  Since then there have been many other works attempting to pin down exactly what trade-offs a random subspace achieves for this problem.

Along the way in the study of list-decodability, several related notions have surfaced.  The first is \em average-radius list-decoding. \em 
\begin{definition}
A code $\cC \subseteq \F^n$ is $(\rho, L)$-average-radius list-decodable if for all sets $\Omega \subseteq \cC$ of size $|\Omega| \geq L$, we have, for all $z \in \F^n$,
\[ \frac{1}{|\Omega|} \sum_{c \in \Omega} d(c, z) \geq \rho. \]
\end{definition}
Notice that average-radius list-decodability is stronger than standard list-decodability, as the average is no larger than the maximum.  However, it turns out to be an extremely useful strengthening.  First, because it is stronger, it is easier to prove lower bounds on average-radius list-decodability than on standard list-decodability, and this is the approach taken in~\cite{gurnar2013}.  Second and more surprisingly, it turns out that it is actually often more natural to prove \em upper \em bounds on average-radius list-decodability, as the sum is nicer to work with than the maximum.  This was the approach taken in~\cite{cgv2012,woot2013,RW14,RW15}.  Further, most natural bounds (both upper and lower) for list-decoding apply to average-radius list-decoding, notably the list-decoding capacity theorem mentioned above.  So, once again we know that, with high probability, a completely random set $\cC \subset \F_q^n$ of size $q^k$ is $(\rho, O(1/\delta))$-average-radius list-decodable with rate $k/n = R^* - \delta$.  The question is whether this applies to random subspaces as well.

Another notion that has been explored recently, both in coding theory and beyond, is that of \em list-recovery. \em 
\begin{definition}
Suppose $\ell \leq |\F|$.
A code $\cC \subset \F^n$ is $(\alpha, \ell, L)$-list-recoverable if for all sets $\Omega \subseteq \cC$ of size $|\Omega| \geq L$ and for all collections of sets $S_1 ,\ldots, S_n \subseteq \F$ of size $|S_i| \leq \ell$, 
\[ \min_{c \in \Omega} \frac{1}{n}\sum_{i =1}^n \ind{ c_i \in S_i } \leq \alpha. \]
\end{definition}
List-recoverability is a generalization list-decodability, where the target $z$ is replaced with a list of sets $S_1,\ldots,S_n$, and distance is defined by $\frac{1}{n} \sum_{i=1}^n \ind{ c_i \not\in S_i}$ rather than by $\frac{1}{n} \sum_{i=1}^n \ind{c_i \neq z_i}$.  That is, $(\alpha, 1, L)$-list-recoverability is the same as $(1 - \alpha, L)$-list-decodability.\footnote{In the literature, it is common for the first parameter in $(\alpha, \ell, L)$-list-recoverability to be agreement, while the first parameter in $(\rho, L)$-list-decodability is disagreement.  We stick with this convention for the introduction; in the technical body of the paper, we will always consider average-radius list-recovery, which captures all of these notions, and we will use the convention that the first parameter is agreement.}
Notice that the question is still interesting when $\alpha = 1$; we refer to this $(1, \ell, L)$-list-recovery as \em $(\ell, L)$ (zero-error) list recovery.  \em  Geometrically, $(\ell, L)$ zero-error list-recovery guarantees that fewer than $L$ points of $\cC$ lie in any $\ell \times \ell \times \cdots \times \ell$ combinatorial rectangle in $\F^n$.
When $\alpha < 1$, the question is the same, but about combinatorial rectangles that have been ``puffed up" by a Hamming ball.

List-recovery was first introduced in the context of list-decoding; many classical list-decoding algorithms (for example, the Guruswami-Sudan algorithm for Reed-Solomon codes~\cite{GS99}) are also list-recovery algorithms, and list-recovery algorithms can be used as building blocks to obtain list-decoding algorithms~\cite{GI02,GI03,GI04a}.  However, list-recovery has since become interesting on its own, and has found numerous applications outside of coding theory~\cite{GI04a,NPR11,INR10, GNPRS13}.
As with list-decoding, there is a list-recovery capacity theorem, and a completely random code $\cC \subseteq \F^n$ achieves it; in this context, the optimal rate is
\[ R^* = \arm{1 - H_{q/\ell}(1 - \alpha) - \log_q(\ell)}. \]  
The natural question, following Question~\ref{q:bigq}, is whether a random linear code of rate approaching $R^*$ can be $(\alpha, \ell, L)$-list-recoverable with reasonable list-size $L$.    In fact, answering this question has immediate applications in coding theory, where the existence of linear list-recoverable codes (over constant or logarithmically-sized alphabet) is a frequent building block: in~\cite{GI01,venkat-thesis}, the authors were not able to obtain a good enough result for the list-decodability of random linear codes, and instead used a family of quasi-linear codes.  In~\cite{HW15}, the authors did use an off-the-shelf result for random linear codes, and this resulted in an exponential increase in their output list size.

\vspace{.5cm}

So far we have discussed several instantiations of Question~\ref{q:bigq}.  We would like some sort of unified answer to all of these questions: not only are the questions interesting and important in coding theory in their own right, but understanding them improves our understanding of Question~\ref{q:bigq}.    All of these questions have been studied before, but there are still gaps in our knowledge.  Moreover, as we discuss below, the approaches that have worked so far are quite varied and do not easily port from one problem to the next.

\paragraph{A disjointed landscape.}
In addition to many different problems (list-decoding, average-radius list-decoding, list-recovery), there are many different parameter regimes.  We may ask for $\rho \to 0$,  which necessitates $R \to 1$ (which is the typical setup in traditional communication setting).  Or we may ask for $\rho \to 1$ and $R \to 0$ (which is useful in computational complexity applications), or anything in between.  We may consider finite fields of size $2$ or $\omega(1)$ or even $\omega(2^n)$.

Broadly speaking, there are three sorts of approaches that have been applied.  The most straightforward, in~\cite{ZP,venkat-thesis}, uses linear independence to mimic the argument in the completely random case.  More formally, if $v_1,\ldots,v_L$ are linearly independent vectors in $\F^k$, and $\mathbf{G} \in \F^{n \times k}$ is a random full-rank matrix, then the vectors $\mathbf{G}v_1,\ldots,\mathbf{G}v_L$ (which are points in our random subspace) are independent, uniformly distributed vectors.  In other words, we may apply the straightforward logic of random codes to the random linear case.  However, the downside of these approaches is that they result in exponentially large list sizes: to guarantee the existence of $L$ linearly independent vectors in a set $\Lambda \subseteq \F^k$, this set $\Lambda$ must have size at least $q^L$.   

A second approach is that of~\cite{GHK11}, who used the linear-algebraic structure of these sets $\Lambda \subseteq \F^k$ in order to get a handle on things.  In a bit more detail, their main technical lemma states that if $v_1,\ldots,v_L$ are chosen at random from a Hamming ball, then their span is unlikely to be contained in that ball.   This is a beautiful argument, but it works only for (standard) list-decoding: it is not clear (to us) how to make it work for list-recovery.  Moreover, it only works in some parameter regimes: where $q$ is constant, and where the disagreement $\rho$ is bounded away from $1 - 1/q$.   As $\rho$ approaches $1 - 1/q$, say $\rho = 1 - 1/q - \delta$, then list size guaranteed by the approach depends exponentially on $\delta$, while a completely random code would obtain a polynomial dependence.

A final line of work uses ideas from high-dimensional probability theory~\cite{cgv2012,woot2013,RW14}.  These arguments bring powerful tools from stochastic processes to bear in coding theory, but again they have their drawbacks.  First, these tools seem too coarse to narrow in on the precise rate: the rate guarantees are off by at best a large constant factor, which makes these results only interesting when the rate is very small.\footnote{That is, the difference between $\delta$ and $\delta/100$ is not as important, if $\delta \to 0$ and we care about the dependence on $\delta$ elsewhere; however, the difference between $1 - H_q(\rho)$ and $(1 - H_q(\rho))/100$ is much more important when $\rho$ is constant and, say, $1 - H_q(\rho) = 1/2$.} 
When the target rate $R^*$ is very small, then the disagreement $\rho$ should be very large, like $\rho = 1 - 1/q - \delta$.  Thus, this approach is effective in a disjoint parameter regime from that of~\cite{GHK11} discussed above.  Second, again these tools work best for small $q$~\cite{cgv2012,woot2013}, although there are results for large $q$~\cite{RW14}.  Third, although these tools do say something for list-recoverability, they do not give the correct dependence on the input list size $\ell$.  

Thus, these three approaches work in fundamentally different ways.  The first works in all parameter regimes, but always results in a list size that is exponential in the gap-to-capacity (that is, the difference between the rate and $R^*$).  The second can alleviate this exponential dependence, but only when the radius $\rho$ is bounded away from $1 - 1/q$.  The third can also alleviate this dependence, but only when $\rho$ approaches $1 - 1/q$.  Moreover, the last two approaches do not extend well to list-recovery.

\paragraph{This work: tying it together.}
In this work, we present a \em single argument \em  that works in all of the aforementioned parameter regimes, and additionally works for list-recovery and average-radius list-decoding.  In fact, we introduce a new concept, \em average-radius list-recovery, \em which generalizes both.
The main downside with our approach is that it results in list sizes which, while independent of $n$, are quasipolynomial in the parameters of interest, rather than polynomial (which would be optimal).  However, our approach is still able to answer several instantiations of Question~\ref{q:bigq}.  We discuss three such applications in detail in Section~\ref{sec:results}.  They include:
\begin{itemize}
	\item \textbf{Improved results for list-decoding and list-recovery of low-rate codes over large alphabets.}
	In this setting, previous achievable rates were off from the optimal $R^*$ by polylogarthmic factors, and depended on $q$ and $\ell$.  Our new results are off by only a factor of $0.99$ (or more generally by $1 -\gamma$ for any constant $\gamma > 0$), and have no dependence on $q$ or $\ell$.
	\item \textbf{Improved list sizes for high-rate list-recovery of random linear codes.}  Previous results for high rate random linear codes could only establish list sizes exponential in $\ell$; our results can reduce this dependence to quasipolynomial.  To the best of our knowledge, our result gives the first \em linear \em high-rate codes which are list-recoverable with output list size that depends sub-exponentially on the input list size $\ell$.
	\item \textbf{Average-radius list-decoding for constant-rate codes over small alphabets.}  By combining the results of~\cite{GHK11} and \cite{RW15}, it was already known that codes of rate $R^* - \delta$ were $(\rho, L)$-average-radius list-decodable with small list sizes, when $\rho$ is bounded away from $1 - 1/q$.  Our argument can match this result (albeit with slightly worse list sizes).  While we do not establish any new results here, we highlight this case to demonstrate that our argument is effective in many different parameter regimes: the approach of~\cite{GHK11} does not seem to generalize to the previous two settings.
\end{itemize}

The outline of the argument is very simple, and we sketch it below in Section~\ref{sec:outline}.  First, we close this section with a roadmap and some basic notation that we will need.  

\paragraph{Outline.}
In Section~\ref{sec:outline} below, we give the high-level idea of our approach.  This approach results in a broad main theorem, which we can specialize in many different ways.
Before going into the technical details of the proof or the main theorem, we discuss in Section~\ref{sec:results} these specializations, and how they relate to previous work.
Finally, to get into the technical meat, we lay the formal groundwork for our approach in Section~\ref{sec:notation} and prove a few basic lemmas that we will need.  
We begin with a warm-up argument in Section~\ref{sec:easy}.  This section presents a special case of our argument for zero-error list-recovery.  This statement is subsumed by our more general theorem, but it highlights the structure of the argument; we encourage the reader to look at this section first for intuition.  Finally, in Section~\ref{sec:hard}, we prove our more general theorem about average-radius list-recoverability.  This requires a bit more care and a lot more notation than the argument in Section~\ref{sec:easy}, but the basic insight is the same.

\paragraph{Notation.}  We briefly define some notation here, which we will use throughout the paper.  (We will set up additional notation that we will need for our proofs in Section~\ref{sec:notation}).  We use $\lg()$ and $\log()$ to denote the logarithm base $2$ and $e$ respectively, and $\exp_q(x) := q^x$.
For a prime power $q$, we will denote by $\F_q$ the finite field with $q$ elements; when the size of the field is clear or does not matter, we will refer to this is $\F$.  As usual, $\EE$ and $\mathbb{P}$ will denote expectation, and probability respectively.
  For a set $X \subset \F^d$, $\spn(X)$ denotes the linear span $\inset{ \sum_{x \in X} \alpha_x x \suchthat \alpha_x \in \F \forall x }$ and for a subspace $V \subseteq \F^d$, $\dim(V)$ refers to the dimension of $V$.   
The volume of the Hamming ball
$\inset{ x \in \F^n \suchthat \dist(x,0) \leq \rho }$ is denoted by $\vol_{|\F|}^n(\rho)$.
Throughout, $C,C',C''$ and so on will denote constant factors, independent of all the relevant parameters unless otherwise stated.  We will overload these constants $C, C', C''...$ periodically, although not within a given scope.
We will use $H_q(x)$ to denote the $q$-ary entropy,
\begin{equation}
\label{eq:entropydef}
 H_q(x) := x \log_q(q - 1) - x\log_q(x) - (1 - x)\log_q(1 - x).
\end{equation}

\section{Overview of approach}\label{sec:outline}
In this section, we give a brief overview of our approach.  
Suppose that $\cC\subseteq \F_q^n$ is a random linear code of dimension $k$; we may write
\begin{equation}\label{eq:linearcode}
 \cC= \inset{ \mathbf{G} v \suchthat v \in \F_q^k }
\end{equation}
for a random matrix $\mathbf{G} \in \F_q^{n \times k}$.\footnote{Strictly speaking, this isn't quite a random subspace of dimension $k$, since the rank of $\mathbf{G}$ might be less than $k$.  However, this event is unlikely and we will ignore it for now. Also note that our generator matrix $\mathbf{G}$ has dimension $n\times k$ instead of the $k\times n$ which is traditional in some communities.}
Recall that we want to understand, broadly speaking, how ``clustered" any set of $L$ elements of $\cC$ is.  For list-decoding, we want to know if any $L$ lie in a ball of a particular radius; for zero-error list-recovery, if any $L$ lie in a combinatorial rectangle; for average-radius list-decoding, how ``close" are any $L$ to a central point $z$, in some averaged sense.   We will define a notion called \em average-radius list-recovery, \em which subsumes all of these notions.  
\begin{definition}\label{def:avgradlr}
Suppose that $\ell < q$ are integers, and $q$ is a prime power.
A code $\cC \subset \F_q^n$ is $(\eps, \ell, L)$-average-radius list-recoverable if, for all sets $\Omega \subset \cC$ with $|\Omega| \geq L$ and for all sets $S_1,\ldots,S_n \subset \F_q$ with $|S_i| \leq \ell$,
\[ \frac{1}{|\Omega|} \sum_{c \in \Omega} \frac{1}{n} \sum_{i=1}^n \ind{ c_i \in S_i } \leq \eps. \]
\end{definition}
It is easy to see that $(\eps, \ell, L)$-average-radius list-recoverability implies $(\eps, \ell, L)$-list-recoverability, and that $(\eps, 1, L)$-average-radius list-recoverability implies $(1 - \eps, L)$-average-radius list-decodability and hence $(1 - \eps, L)$-list-decodability.  Our main theorem establishes average-radius list-recoverability, and hence can establish all of these other notions.

Suppose that $\cC$ has the form \eqref{eq:linearcode}, so that a codeword $c \in \cC$ has the form $c = \mathbf{G}v$ for some $v \in \F_q^k$, and the $i$'th symbol of $c$ is $c_i = \ip{g_i}{v}$, where $g_i$ is the $i$'th row of $\mathbf{G}$.  In this case, the condition in Definition~\ref{def:avgradlr} reads that for all $S_1,\ldots, S_n \subset \F_q$ of size at most $\ell$, for all sets 
$\Lambda \subseteq \F_q^k$ of size $L$,
\begin{equation}
\label{eq:goodthing}
 \sum_{v \in \Lambda} \sum_{i = 1}^n \ind{\ip{ g_i }{ v } \in S_i } \leq \eps L n,
\end{equation}
That is, there should be no sets of messages $\Lambda$ of size $L$ that agree too much, on average, with the lists $S_i$.  
Thus, our goal will be to show that \eqref{eq:goodthing} holds for all $S_1,\ldots,S_n$, and for all $\Lambda$.

Our starting point is the approach of Zyablov and Pinsker~\cite{ZP} discussed above.  Their observation was 
 that if the underlying message vectors $\Lambda = \inset{v_1,\ldots, v_L}$ happen to be linearly independent, then there is no problem: the encodings $\mathbf{G}v_1,\ldots, \mathbf{G}v_L$ under a random linear map $\mathbf{G}$ are independent random variables, and things work as expected.  At this point their argument ended: if the list size $|\Lambda|$ is exponentially large, then $\Lambda$ must contain some large linearly independent subset, and we are done.
 
 However, we wish to avoid the large list size.  We know from the discussion above that if the underlying set of messages $\Lambda$ is high-dimensional, then we are okay, and \eqref{eq:goodthing} holds.  The problem arises if $\Lambda$ contains many linear dependencies.   
 
 Our main idea is to take advantage of such dependencies to run a recursive argument: if indeed $\Lambda$ is ``problematic," meaning that \eqref{eq:goodthing} fails to hold for some sets $S_1,\ldots,S_n$, then $\Lambda$ has many linear dependencies.  We show that, when formalized in the correct way, this means that $\Lambda$ contains a large, low-dimensional subset $\Lambda'$.   Consider projecting $\Lambda$ onto the space spanned by $\Lambda'$, and also projecting the rows $g_i$ of the matrix $\mathbf{G}$ onto this space to obtain rows $g_i'$.  
 The sets of inner products 
 \[ \inset{\ip{ g_i }{ v } \suchthat v \in \Lambda } \]
 for $i = 1, \ldots, n$ are related to the sets
 \[ \inset{ \ip{g_i'}{v'} \suchthat v' \in \Lambda' }. \]
 Indeed, the latter are just subsets of the former, if the projection is done in the right way.  This suggests that if $\Lambda$ was problematic, then $\Lambda'$ is also problematic, since the condition \eqref{eq:goodthing} depends only on these inner products.  This transitivity of problematic-ness is not obvious, but turns out to be true.  
 
 At this point, we may recurse: if $\Lambda'$ is problematic, then just as with $\Lambda$, it has some linear structure at fault, and we may exploit this to find a large, even lower-dimensional subset $\Lambda''$.   We continue this way until we eventually find a subset $\Gamma$ of dimension $d$ and of size larger than $q^d$.  This is a contradiction, as we are always dealing with sets, rather than multisets.  So we conclude that there was no problematic $\Lambda$ to begin with.
 
 The main technical challenge in implementing the above idea is finding the right way to quantify ``problematic" and ``linear structure."  We will use a quantity which we call $\sigma_p(\Lambda)$:
 \[ \sigma_p(\Lambda) := \EE_{v_1,\ldots, v_p \in \Lambda} q^{-\dim(v_1,\ldots,v_p)}, \]
 where above the expectation is over $v_1,\ldots, v_p$ drawn uniformly at random, with replacement, from $\Lambda$.  
 Intuitively, this is some measure of the linear dependencies in $\Lambda$.  If $\Lambda$ were very low-dimensional, or contained a large low-dimensional subset, then we would likely draw many linearly dependent vectors, and $q^{-\dim(v_1,\ldots,v_p)}$ would be not too small.  On the other hand, if $\Lambda$ had very little linear structure, then we might expect to usually draw $p$ linearly independent vectors, and this quantity would be very small.  
 
 At first glance, this may seem like a strange measure of linear dependency; we briefly explain why it is helpful.  Consider \eqref{eq:goodthing}.  By switching the order of the sums, we may write this as
 \[ \sum_{v \in \Lambda} \sum_{i=1}^n \ind{\ip{g_i}{v} \in S_i } 
 = \sum_{i=1}^n X_i,
 \]
 where
 \[ X_i = \sum_{v \in \Lambda} \ind{ \ip{g_i}{v} \in S_i}.\]
 Notice that, if the $g_i$ are rows of a random generator matrix $\mathbf{G}$, then the $X_i$ are all independent random variables.  This means that the quantity we are trying to control is the sum of independent random variables.  This immediately suggests some sort of Chernoff bound; unfortunately, the naive application of such a bound will fail, because the $X_i$ themselves may not be very concentrated.
 
 The crux of the argument is that, if there is not much linear structure in $\Lambda$---that is, if $\sigma_p(\Lambda)$ is small for all $p$---then in fact the $X_i$ are reasonably concentrated, and the Chernoff approach will work.  To see this, consider the $p$'th moments of the $X_i$.  For one $X$ (we drop the $i$ subscript for clarity) we have
{\allowdisplaybreaks
 \begin{align*}
 \EE_g X^p &= \EE_g \inparen{\sum_{v \in \Lambda} \ind{ \ip{g}{v} \in S}}^p \\
 &= \EE_g \sum_{v_1,\ldots,v_p \in \Lambda} \prod_{j} \ind{\ip{g}{v_j} \in S_i} \\
&= \sum_{v_1,\ldots,v_p \in \Lambda} \mathbb{P}_g\inset{\ip{g}{v_j} \in S \ \forall j \in [p]},
 \end{align*}
}
 where in the final line we have switched the order of the sum and the expectation.  Now, this probability $\mathbb{P}_g\inset{ \ip{g}{v_j} \in S\  \forall j}$ is related to the dimension of $\inset{v_1,\ldots,v_p}$.  More precisely, suppose that $S = \inset{0}$; then this would just be
 \[ \sum_{v_1,\ldots,v_p \in \Lambda} \mathbb{P}_g \inset{\ip{g}{v} = 0 \ \forall j} =  \sum_{v_1,\ldots,v_p \in \Lambda} q^{-\dim(v_1,\ldots,v_p)} = L^p \sigma_p(\Lambda). \]
 In the case where $S \neq \inset{0}$, we must adjust the calculation slightly, but hopefully it is now clear that the $\sigma_p(\Lambda)$ play a natural role in bounding the moments of the $X_i$.  
 
 Thus, picking up our earlier line of thought, we see that if there is not too much linear structure in $\Lambda$---so $\sigma_p(\Lambda)$ is small for all $p$---then the moments of the $X_i$ are all small.  This means that we can apply a Chernoff-like analysis to bound $\sum_i X_i$ above, and conclude that $\Lambda$ will not pose a problem for establishing \eqref{eq:goodthing}.  The final ingredient is a (straightforward) argument that if $\sigma_p(\Lambda)$ is large for some $p$, then we can extract a large, low-dimensional subset $\Lambda' \subseteq \Lambda$.  Then we may carry out the argument as described above.

\arm{We note that our argument has the same structure as an {\em energy increment argument}, first used by Roth~\cite{roth}.}

 While the outline of the argument is quite simple, there are several complications.  The first is in the quantitative details in the Chernoff-like argument, which requires reasonably careful control of the moments.  The second is the issue we alluded to earlier, that it is not immediately obvious that our definition of ``problematic" translates from $\Lambda$ to the large, low-dimensional subset $\Lambda'$.  This again requires a somewhat delicate argument.  
 
 These complications seem necessary for our general results, but they can both be significantly simplified in the case of zero-error list-recovery.  In this case, it turns out that the Chernoff-like bound can be replaced with a simple application of Markov's inequality.   Moreover, the transitivity of ``problematic-ness" does turn out to be obvious when our definition of ``problematic" is given by zero-error list-recovery, rather than average-radius list-recovery.  While this is a restrictive case, this allows for the argument to go through much more cleanly.  To that end, we give a much easier ``warm-up" argument covering this case in Section~\ref{sec:easy}.  This argument is subsumed by our final argument in Section~\ref{sec:hard}, but it highlights all of the main ideas.  We recommend that the reader interested in the technical details begin with the warm-up argument before diving into the full argument.
 
 Before we say any more about the argument, however, in the next section we outline three corollaries of this approach, and survey how they relate to the literature on list-decoding and list-recovery.

\section{Results and related work}\label{sec:results}

Our main theorem (Theorem~\ref{thm:avgrad}) is general and can be specified to different problems in list-decoding and list-recovery.  We defer the statement of Theorem~\ref{thm:avgrad} to Section~\ref{sec:hard}, but here we give three parameter regimes in which special cases of Theorem~\ref{thm:avgrad} are interesting, and state these results as Corollaries~\ref{cor:largeq}, \ref{cor:highratelr}, and \ref{cor:constantagr}.  The proofs of these corollaries appear in Section~\ref{ssec:cors}, after the statement of Theorem~\ref{thm:avgrad}.

\subsection{List-decoding and list-recovery over large alphabets}
We first consider the setting where the rate of the code is small, and the alphabet size is large.  This is the setting studied in~\cite{cgv2012,woot2013,RW14} and is relevant for related notions in pseudorandomness, like expanders, extractors, and hardness amplification~\cite{vadhan2011}.
More precisely, for a large alphabet size $q$, we consider the problem of $(1 - \eps,L)$-list-decodability, or $(\eps,\ell,L)$-list-recovery, when $\eps = 1/q + \delta$, for some small $\delta$ and for $q \gg delta^{-1}$.  Here, the optimal rate (given by the list-decoding capacity theorem) is on the order of $\delta$ for both list-decoding and list-recovery (when $q \geq \poly(\ell)$ is sufficiently large) and the ideal list size is polynomial in $1/\delta$ or $\ell/\delta$, respectively. 

In this parameter regime, the best results are given by~\cite{RW14}.  While the list sizes in that work are about right, the rates are suboptimal.  For list-decoding, the rate given in~\cite{RW14} is on the order of $\delta / (\polylog(1/\delta)\log(q))$, instead of $\Omega(\delta)$.  When $q$ is growing, this tends to zero; and even when $q$ is constant, the dependence on $\delta$ (not to mention the leading constants) are not correct.
For list-recovery, the situation is even worse.  It is easy to see that any result which guarantees a rate $R$ code for average-radius list-decoding can also guarantee a rate $R/\ell$ code for list-recovery.  Thus,
the work of~\cite{RW14} implies that a random linear code of rate $\tilde{O}(\delta/\ell)$ is $(1/q + \delta, \ell, L)$-average radius list-recoverable---however, as mentioned above, the highest achievable rate should not depend on $\ell$.
To summarize, the major open question in this regime is:
\begin{question}
Over large alphabets,
is a random rate $\Omega(\delta)$ linear code  $(1 - 1/q - \delta, L)$-average-radius list-decodable, or $(1/q + \delta, \ell, L)$-average-radius list-recoverable, for $L = \poly(\ell/\delta)$, with high probability?
\end{question}

Theorem~\ref{thm:avgrad} answers this question---moreover, not only can we get rate $\Omega(\delta)$ (removing the dependence on $\ell$ and $q$, as well as the extra logarithmic factors) but we may make the constant inside the $\Omega(\cdot)$ arbitrarily close to the optimal constant. 
More precisely, we can prove the following corollary.
\begin{corollary}[List-decoding and list-recovery over large alphabets]
	\label{cor:largeq} 
	Let $\ell\ge 1$ be an integer.
	For every sufficiently small constant $\gamma > 0$, there are constants $C, C'$ (which depend on $\gamma$) so that the following holds.
	Choose $\delta > 0$ sufficiently small.
	Let $q \geq \max \inset{C(\ell/\delta)^2, \ell^{C/\arm{\delta}}}$.
	Suppose that
	\[ R \leq \inparen{1 - H_{q/\ell}\inparen{1 - \frac{\ell}{q} - \delta} - \log_q(\ell)}(1 - \gamma). \]
	Then a random linear code of rate $R$ over $\F_q$ is $(\eps, \ell, L)$-average-radius list-recoverable with high probability, for
	\[ \eps = \frac{\ell}{q} + \delta \]
	and 
	\[ L  \leq q^{ C' \log^2(\ell/\delta) } . \]
\end{corollary}

We summarize existing results, as well as our results, for both list-decoding and list-recovery in this parameter regime, in Table~\ref{table:largeq}.
\begin{table}[ht!]
\begin{center}
\begin{tabular}{|c|c|c|}
\hline
Source & $R$ & $L$ \\
\hline
\hline
Uniformly random & $1 - H_q( 1 - 1/q - \delta ) - \xi$ & $O(1/\xi)$ \\
\hline
\cite{ZP} & $1 - H_q(1 - 1/q -\delta) - \xi$ & $q^{O(1/\xi)}$ \\
\hline
\cite{RW14} & $\Omega\inparen{ \frac{ \delta }{ \log^5(1/\delta) \log(q) } }$ & $O(1/\delta)$ \\
\hline
This work & $0.99 \cdot \inparen{1 - H_q(1 - 1/q - \delta)}$ & $q^{O(\log^2(1/\delta))}$ \\
\hline
\multicolumn{3}{c}{List Decoding}
\end{tabular}
\hspace{.1cm}
\begin{tabular}{|c|c|c|}
\hline
Source & $R$ & $L$ \\
\hline
\hline
Uniformly random & $\arm{1 - H_{q/\ell}(1 - \ell/q - \delta)  - \log_q(\ell) } - \xi $  & $O(\ell/\xi)$ \\
\hline
\cite{venkat-thesis} &  $\arm{1 - H_{q/\ell}(1 - \ell/q - \delta)  - \log_q(\ell) } - \xi $ &  $\inparen{\frac{q}{\ell}}^{O(\ell/\xi)}$ \\
\hline
\cite{RW14} & $\Omega\inparen{ \frac{1}{\ell} \cdot \frac{ \delta }{ \log^5(\ell/\delta) \log(q) } }$ & $O(\ell/\delta)$ \\
\hline
This work & $0.99 \cdot  \inparen{1 - H_{q/\ell}(1 - \ell/q - \delta) - \log_q(\ell) }$ & $q^{ O(\log^2(\ell/\delta) ) } $ \\
\hline
\multicolumn{3}{c}{List Recovery}
\end{tabular}

\end{center}
\caption{Results on list-decodability and list-recoverability for large $q$, when the fraction of agreement $\eps$ is $\ell/q + \delta$ for small $\delta$.  The first table shows results for list-decoding when $q = \Omega(\delta^{-2})$; the second table shows results for list-recovery when $q \arm{\ge \poly\inparen{\ell^{1/\delta}} }$.  For all $\ell \geq 1$ (that is, both list-decoding and list-recovery) in these parameter regimes we have $H_{q/\ell}(1 - \ell/q - \delta) = \Theta(\delta)$.  In both tables, $\xi > 0$ is a parameter indicating how much we back off from the optimal rate; $\xi$ must be $\xi = O(\delta)$ for the rate to be positive.  Our bounds (and those of~\cite{RW14}) do not have a $\xi$ in them---for these, it is implicitly required that $\xi = \Omega(\delta)$.  Note that the bounds of~\cite{ZP,venkat-thesis} are exponential in $1/\delta$ in this setting, so Corollary~\ref{cor:largeq} gives an improvement from exponential to quasipolynomial.
} 
\label{table:largeq}
\end{table}
Our bounds on the rate are nearly tight.  The weakness is that we might hope for the rate to be of the form 
\begin{equation}
\label{eq:rate4}
\arm{1 - H_{q/\ell}(1 - \ell/q - \delta)- \log_q(\ell) - \xi}
\end{equation}
where $\xi \to 0$, possibly much faster than $\delta$.  However, our results in Corollary~\ref{cor:largeq} would require $\xi = \Omega(\delta)$ in \eqref{eq:rate4}.  As we will see in Section~\ref{subsec:ghk-recover}, we can actually attain this sort of dependence when $q$ is constant and the agreement $\eps$ is bounded away from $1/q$.

\subsection{High-rate list recovery}
Our second corollary is for high-rate list-recovery of random linear codes.
This setting is interesting because (to the best of our knowledge) we do not know any \em explicit \em constructions of high-rate, linear, list-recoverable codes.  There are many constructions of high-rate list-recoverable codes (e.g., \cite{GR08, GW13, Kop15, GX13} to name a few).  However, none of these constructions are linear: they all rely on manipulating the code alphabet (folding, adding derivatives, and so on), which destroys linearity.  To that end, when a high-rate linear list-recoverable code is needed (for example, in the constructions of \cite{GI01,HW15}), a random linear code is the best we can do.  In these constructions, such a code is used as an ``inner code" of constant size, so the fact that it is not explicit does not matter.  However, the previous best result on the list-recoverability of random linear codes, which appears in \cite{venkat-thesis}, requires a very large list size, exponentially large in $\ell$.  To that end, \cite{GI01} introduced a family of ``pseudo-linear" codes, which they could show have smaller list size---these had the linearity properties that they needed, although they are not linear.  This was not enough for~\cite{HW15}, who needed high-rate linear list-recoverable codes to instantiate an expander-code construction; in that work, they took the exponential hit in the list size and used the result of~\cite{venkat-thesis}.

Our approach is able to establish better list-recovery of high-rate random linear codes (which immediately improves the list sizes in~\cite{HW15}).  This also improves the state-of-the-art for for \em any \em high rate linear code, as random linear codes are so far our only avenue of attack.  In Corollary~\ref{cor:highratelr}, we are able to reduce the list size of~\cite{venkat-thesis} from exponential in $\ell$ to quasipolynomial.  We summarize the state of the literature in Table~\ref{tab:highratelr}.

\begin{corollary}[High-rate list-recovery] \label{cor:highratelr} 
	There are constants $C, \gamma_0$ so that the following holds.
	Choose $0 < \gamma < \gamma_0$ sufficiently small and let $\ell > 1$ be an integer.
	Suppose that $q \geq  \ell^{C/\gamma}$, and let $\cC$ be a random linear code of rate $R$ for some $R$ satisfying
	\[ R \leq 1 - \gamma. \]
	Then with high probability, $\cC$ is $(1 - \gamma/10, \ell, L)$-list-recoverable, for some
	\[ L \leq \inparen{ \frac{q\ell}{\gamma}}^{\log(\ell)/\gamma} \cdot \exp\inparen{\frac{\log^2(\ell)}{\gamma^3}}. \]
\end{corollary}

\begin{table}[h]
	\begin{center}
\begin{tabular}{|c|c|c|}
\hline
	Source &  Rate & List size  \\ 
	\hline\hline
	Uniformly random code & $1 - \gamma$& $O(\ell/\gamma)$  \\ 
	\hline
Random pseudo-linear code \cite{GI01} & $1 - \gamma$ & $O\inparen{\frac{\ell \log(\ell)}{\gamma^2}}$ \\ 
\hline
Random linear code \cite{venkat-thesis} & $1 - \gamma$ & $\ell^{O(\ell/\gamma^2)}$ \\ 
\hline
Random linear code (this work) & $1 - \gamma$ &  $\inparen{ \frac{q\ell}{\gamma} }^{ \log(\ell)/\gamma^3}$\\ 
\hline
\end{tabular}
\caption{ List-recoverability results for high-rate codes.  While there are several explicit constructions of non-linear high-rate codes, for several applications linear codes are desirable.  We note that the first two rows of the table are not linear codes.   All results listed are for $(1 - O(\gamma), \ell, L)$-list-recovery, and have alphabet size $\ell^{O(1/\gamma)}$.}
\label{tab:highratelr}
\end{center}
\end{table}
\subsection{Average-radius list-decoding over constant alphabets}
\label{subsec:ghk-recover}
The final parameter regime we consider is in some sense the opposite of the first one, where the agreement fraction $\eps$ is bounded away from $1/q$, and where the alphabet size $q$ is constant.
In this setting, we consider (average-radius) list-decoding.
One goal, in the spirit of Question~\ref{q:bigq}, is to show that a random linear code meets the list-decoding capacity theorem.
More precisely, for $\eps > 1/q$ (and bounded away from $1/q$), 
we would like to show that a random linear code of rate
\[ R = 1 - H(1 - \eps) - \delta =: R^* - \delta \]
is $(1 - \eps, L)$-list-decodable for reasonable-sized $L$, as $\delta \to 0$.

As mentioned in the previous section, an argument of Zyablov and Pinsker~\cite{ZP} shows that a random linear code of rate $R$ approaching $R^*$ with list size $L = q^{1/\delta}$.  However, a completely random code of rate nearly $R^*$ achieves $L = O(1/\delta^2)$, and reducing the list size to something sub-exponential in $\delta$ remained open for nearly three decades.  In~\cite{GHK11}, Guruswami, H{\aa}stad, and Kopparty achieved the correct rate and list size, for list-decoding.  
Their proof only holds for (standard) list-decoding, rather than average-radius list-decoding, but by applying an observation made in~\cite{RW15} one can convert their result to average-radius list decoding by taking a polynomial hit in the list size: together, these results establish $(1 - \eps, L)$-average-radius list-decodability for rates $R = R^* - \delta$ and $L = O(1/\delta^3)$.
Cheraghchi, Guruswami and Velingker~\cite{cgv2012} and Wootters~\cite{woot2013} obtained the correct list size of $O(1/\delta^2)$, for average-radius list-decoding, but lost a constant factor in the rate, requiring $R \leq R^*/C$ for some constant $C$; it does not seem like these techniques can be tightened to obtain $C = 1$.  (As mentioned above, these results are most interesting when $\eps \to 1 / q$).




Our approach can recover the optimal rates obtained in~\cite{GHK11} combined with~\cite{RW15}, although with a worse list size, for a large range of parameters.  More precisely, for constant alphabet sizes, we show that a random linear code of rate $R^* - \delta$ is average-radius list-decodable, with only a quasipolynomial dependence on $\delta$. 
This does not give any improvement over the state-of-the-art; however, we include this example (Corollary~\ref{cor:constantagr} below) because it highlights the flexibility of our approach.  The techniques of~\cite{GHK11} do not seem to extend to the parameter regime where the agreement $\eps$ approaches $1/q$.  On the other hand, the techniques of~\cite{cgv2012, woot2013}, which do work in that parameter regime, cannot obtain the current rate for $\eps$ bounded away from $1/q$.  Corollaries~\ref{cor:largeq} and \ref{cor:constantagr} together show that Theorem~\ref{thm:avgrad} can be adapted to do both at once.  We summarize the state of the literature in Table~\ref{table:cor1}.

\begin{corollary}
[Average-radius list-decoding over constant alphabets]
\label{cor:constantagr}
Choose any constant $q \geq 2$.  There are constants $C, \delta_0$, and $\eps_0 \lneq \eps_1$ (which depend on $q$) so that the following is true.
For all $\delta \in (0, \delta_0)$, and for all $\eps \in (\eps_0, \eps_1)$, a random linear code of rate
\[ R = 1 - H_q(1 - \eps) - \delta \]
is $(\eps, L)$-average-radius list-decodable with list size
\[ L \leq \inparen{ \frac{q}{\delta} }^{C \log^2(1/\delta)}. \]

Moreover, for $q = 2$, we may take $\eps_0 = 0.51$ and $\eps_1 = 0.8$, and in general we may take
\[ \eps_0 = 1/q + 1/q^2 , \qquad  \eps_1 = \max \inset{ 0.8,  1 - \inparen{ \frac{ 1.1 \cdot \ln(\arm{q+1}) }{q} }}. \]
\end{corollary}

\begin{table}
\begin{center}
\begin{tabular}{|c|c|c|c|}
\hline
Source &  $R$ & $L$ \\
\hline
\hline
Uniformly random code &  $R^* - \delta $ & $O(1/\delta^2)$ \\
\hline
\cite{ZP} &  $R^* - \delta$ & $\exp(1/\delta)$ \\
\hline
\cite{GHK11}+~\cite{RW15}  & $R^* - \delta$  & $O(1/\delta^3)$ \\ 
\hline
\cite{woot2013} &  $R^*/C$ & $O(1/\delta^2)$ \\
\hline
This work &  $R^* - \delta$ & $q^{C \log^2(1/\delta)}$ \\
\hline
\end{tabular}
\caption{Previous results for $(1 - \eps, L)$ average-radius list-decodability of random linear codes of rate $R$ (with high probability).  Here, $R^* = 1 - H_q(1 - \eps)$ is the optimal rate.  Our results, as well as the results of \cite{GHK11}, require the agreement fraction $\eps$ to be bounded away from $1/q$ by a constant.  Our results additionally require that $\eps$ be bounded away from $1$ by a constant.}
\label{table:cor1}
\end{center}
\end{table}

\section{Preliminaries}
\label{sec:notation}
Throughout, we will be interested in \em linear codes \em $\cC$ of length $n$ over a field $\F$ of size $q$.  That is, $\cC \subseteq \F^n$ is a linear subspace of $\F^n$.  A natural definition of a \em random linear code \em is a random subspace of $\F^n$.  However, as with other works, we choose a slight tweak on this definition which is a bit easier to work with.
%
\begin{definition}\label{def:moreconvenient}
	For $R \in (0,1)$ and $n \in \mathbb{N}$ so that $Rn \in \mathbb{N}$, a \em random linear code \em $\cC \subset \F^n$ of rate $R$ is the set
	\[ \inset{ \bX \cdot v \suchthat v \in \F^{Rn} }, \]
	where $\bX$ is a matrix whose rows are chosen independently and uniformly at random from $\F^{Rn}$.
\end{definition}
We note that, when defined this way, 
a ``random linear code of rate $R$" might not have rate $R$, in the sense that $\dim(\cC)$ may be less than $Rn$ if the matrix $\bX$ is not full rank.
However, the probability that $\bX$ is not full rank is exponentially small in $n$, and so for all of the results in this work (which hold with high probability), the distinction does not matter.  In particular, all results hold for the definition of a random linear code as a uniformly random $k$-dimensional subspace of $\F_q$.

For some $d \leq \dim(\cC)$, let $X = \inset{x_1,\ldots, x_n} \subseteq \F^d$ be a full-rank set of size $n$.
We will use $\cols(X) \subseteq \F^n$ to refer to the set
\[ \cols(X) = \inset{ ( (x_1)_j, (x_2)_j, \ldots, (x_n)_j ) \suchthat j \in [d] }. \]

Suppose that $\cC$ is a linear code, and consider a set $X$ so that $\cols(X) \subseteq \cC$.  Because $\cC$ is linear, any linear combination of the vectors in $\cols(X)$ will lie in $\cC$ as well.  Another way to write this is to say that for all $v \in \F^d$, the vector
\[ c =  ( \ip{x_1}{v}, \ip{x_2}{v} , \ldots, \ip{x_n}{v} ) \]
has $c \in \cC$.  Given a set $X$, a set $\Lambda \subseteq \F^d$ of size $L$ specifies a set of codewords
\[\inset{  ( \ip{x_1}{v}, \ip{x_2}{v} , \ldots, \ip{x_n}{v} )\suchthat v \in \Lambda}. \]
In this work, we will rely heavily on the linear structure of such sets $X$ and $\Lambda$, using the above correspondence between sets of codewords and pairs $(X,\Lambda)$.

Throughout this paper, a set $X$ will always be a full-rank set of size $n$, and a pair of $X$ and $\Lambda$ should always be thought of in the context above: that is, $(X,\Lambda)$ is as a description of a set of codewords.  Since we are interested in list-recoverability, we will be interested in controlling ``bad" sets of codewords; this amount to controlling ``bad" pairs of $(X,\Lambda)$. 

More precisely, 
(as is also done in~\cite{RW14}), we consider the \em relative pluralities \em of a set $\Lambda$ with respect to a vector $x \in \F^d$ (which should be thought of as an element of $X$).  For a vector $x \in \F^d$ and a set $\Lambda \subseteq \F^d$, we write
\[ \pl_x(\Lambda) := \frac{1}{|\Lambda|} \max_{\alpha \in \F}\inabs{ \inset{ v \in \Lambda \suchthat \ip{x}{v}=\alpha } }. \]
In particular, if $x = x_i$ is the $i$'th element of some set $X$, and $\Lambda$ is any set, then $\pl_{x_i}(\Lambda)$ is the share of the most frequently-occurring symbol in the $i$'th position in codewords represented by $X$ and $\Lambda$.
It is not hard to see (and we formalize this later in Proposition~\ref{prop:plrs}) that average-radius list-decodability of a code whose generator matrix $\bX$ has rows $X$ is equivalent to the condition that $\sum_{x \in X} \pl_x(\Lambda)$ is small for all sets $\Lambda$.

In this work, we study average-radius \em list-recovery \em (Definition~\ref{def:avgradlr}) which as we noted above is stronger than both list-recovery and average-radius list-decodability.
To reason about average-radius list-recovery, we introduce the notion of the {\em top-$\ell$ relative plurality}. 
 For a set $\Lambda \subset \F^d$, an integer $\ell\ge 1$, and a vector $x \in \F^d$, we define the  \mrk{top-$\ell$} relative plurality  of $\Lambda$ with respect to $x$ to be
\[ \mrk{\pl^{(\ell)}_x}( \Lambda ) := \frac{1}{|\Lambda|} \mrk{\sum_{i=1}^{\ell} \tp_i(x,\Lambda)},\]
\mrk{where $\tp_i(x,\Lambda)= \inabs{ \inset{ v \in \Lambda \suchthat \ip{x}{v} = \alpha } }$, where $\alpha\in\F$ is the value with the $i$th largest $\inabs{ \inset{ v \in \Lambda \suchthat \ip{x}{v} = \alpha } }$. We will also refer to this value $\alpha$ by $\argtp_i(x,\Lambda)$.}
Just as with pluralities, if $x = x_i$ is the $i$'th element of $X$, $\pl_{x_i}^{(\ell)}(\Lambda)$ is the share of the most-popular $\ell$ symbols in the $i$'th position for codewords represented by $X$ and $\Lambda$.
Just as the sum of the pluralities captures average-radius list-decodability, we will see in Proposition~\ref{prop:plrs} that the sum of the top-$\ell$ pluralities capture average-radius list-recovery.
The picture that the reader should have in mind for $X, \Lambda$ and the definition of $\pl_x^{(\ell)}(\Lambda)$ is illustrated in Figure~\ref{fig:defs}.

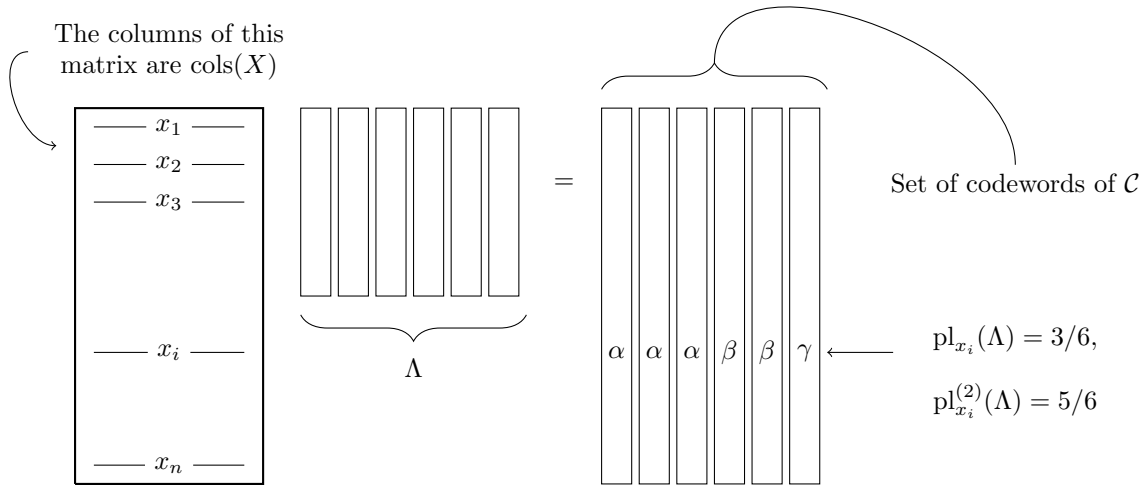
\begin{figure}
\begin{center}
\begin{tikzpicture}[scale=.5]
\draw[thick] (0,0) to (5,0) to (5,10) to (0,10) to (0,0);
\coordinate(a0) at (.5, 9.5);
\coordinate(b0) at (4.5, 9.5);
\node(x0) at (2.5, 9.5) {$x_1$};
\draw (a0) -- (x0) -- (b0);
\begin{scope}[yshift=-1cm]
\coordinate(a0) at (.5, 9.5);
\coordinate(b0) at (4.5, 9.5);
\node(x0) at (2.5, 9.5) {$x_2$};
\draw (a0) -- (x0) -- (b0);
\end{scope}
\begin{scope}[yshift=-2cm]
\coordinate(a0) at (.5, 9.5);
\coordinate(b0) at (4.5, 9.5);
\node(x0) at (2.5, 9.5) {$x_3$};
\draw (a0) -- (x0) -- (b0);
\end{scope}
\begin{scope}[yshift=-6cm]
\coordinate(a0) at (.5, 9.5);
\coordinate(b0) at (4.5, 9.5);
\node(x0) at (2.5, 9.5) {$x_i$};
\draw (a0) -- (x0) -- (b0);
\end{scope}
\begin{scope}[yshift=-9cm]
\coordinate(a0) at (.5, 9.5);
\coordinate(b0) at (4.5, 9.5);
\node(x0) at (2.5, 9.5) {$x_n$};
\draw (a0) -- (x0) -- (b0);
\end{scope}
\foreach \i in {1,...,6}
{
	\draw (5 + \i, 10) rectangle (5.8 + \i , 5);
}
\draw [decorate,decoration={brace,amplitude=10pt,mirror}]
(6, 4.5) -- (12, 4.5)
node [black,midway,yshift=-20pt] {$\Lambda$};
\node at (13, 8) {$=$};
\foreach \i in {1,...,6}
{
	\draw (13 + \i, 10) rectangle (13.8 + \i , 0);
}
\node at ( 14.4, 3.5) {$\alpha$};
\node at ( 15.4, 3.5) {$\alpha$};
\node at ( 16.4, 3.5) {$\alpha$};
\node at ( 17.4, 3.5) {$\beta$};
\node at ( 18.4, 3.5) {$\beta$};
\node at ( 19.4, 3.5) {$\gamma$};

\node(b) at (25, 8) {{Set of codewords of $\cC$}};
\node(c) at (25, 3.5) {\begin{minipage}{3cm} \[\pl_{x_i}(\Lambda) = 3/6, \]\[\pl_{x_i}^{(2)}(\Lambda) = 5/6\]\end{minipage}};
\draw [decorate,decoration={brace,amplitude=10pt}]
(14, 10.5) -- (20, 10.5)
coordinate[midway,yshift=10pt](a) ;
\draw (a) to [out=90,in=90] (b);
\draw[->] (c) to (20, 3.5);
\node(d) at (2.5,11.5) {\begin{minipage}{3.5cm} \begin{center} The columns of this matrix are $\cols(X)$ \end{center} \end{minipage}};
\draw[->] (d) to[out=180,in=180] (-.5,9);
\end{tikzpicture}
\end{center}
\caption{Illustration of how $X = \inset{x_1,\ldots, x_n} \subseteq \F^d$ and $\Lambda \subseteq \F^d$ will be used in this paper.}
\label{fig:defs}
\end{figure}

For a vector $z \in \F^n$ and a vector $y \in \F^n$, $\dist(y,z)$ will denote the relative Hamming distance:
\[ \dist(z,y) := \frac{1}{n}\sum_{i =1}^n \ind{ y_i \neq z_i }. \]
We extend this in the natural way for use with list-recovery: for a matrix $z \in \F^{\ell \times n}$ (which we view as a collection of $n$ lists of size $\ell$), and a vector $y \in \F^n$, we overload notation and write
\[ \dist(z,y) := \frac{1}{n} \sum_{i=1}^n \ind{ y_i \not\in \cup_{j=1}^\ell \inset{ z_{j,i} } }. \] 

As mentioned above, in order to control sets of codewords, we will instead control pairs $(X,\Lambda)$.
The linear structure will be important, and we will get a handle on it by repeatedly projecting $X$ and $\Lambda$ down to lower-dimensional spaces.
If a set of codewords is defined by $X$ and $\Lambda$, and $\Lambda$ happens to be low-dimensional, then the same set of codewords can be represented by $X'$ and $\Lambda'$ which live in a lower-dimensional space.  We make this precise with the following lemma.

\begin{lemma}\label{lem:projection}
	Let $X = \inset{x_1,\ldots,x_n} \subset \F^d$ and $\Lambda = \inset{ v_1,\ldots,v_L} \subset \F^d$ so that $\dim(\spn(\Lambda)) = d' \leq d$.
	Then there is a set $X' = \inset{x_1', \ldots, x_n'}  \subseteq \F^{d'}$, and a set $\Lambda' = \inset{ v_1', \ldots,v_L'} \subseteq \F^{d'}$ so that for all $i \in [L],j\in[n]$,
\[ \ip{v_i}{x_j} = \ip{v_i'}{x_j'}. \]
Further, $X'$ is of the form $X' = \inset{ Ax \suchthat x \in X}$ for some matrix $A$ that depends only on $\Lambda$ (and not $X$).
In particular:
\begin{itemize} 
	\item The (ordered) set $\inset{ \pl_x^{(\ell)}(\Lambda) \suchthat x \in X}$ is preserved when we replace $X$ with $X'$ and $\Lambda$ with $\Lambda'$.  
	\item $\cols(X') \subseteq \spn(\cols(X))$
	\item If $X$ is full-rank then $X'$ is full-rank. 
\end{itemize}
\end{lemma}
\begin{proof}
	Let $P \in \F^{d' \times d}$ be a matrix that maps $\spn(\Lambda)$ injectively onto $\F^{d'}$: that is, $\F^{d'} = \inset{ Pv \suchthat v \in \spn(\Lambda) }$.  
	Let $P^\dagger \in \F^{d \times d'}$ be the pseudo-inverse of $P$ so that in particular $P^\dagger P v = v$ for all $v \in \spn(\Lambda)$.  Then let $X' = \inset{ (P^\dagger)^Tx \suchthat x \in X}$ and let $\Lambda' = \inset{ Pv \suchthat v \in \Lambda}$.  We have, for all $x \in X, v \in \Lambda$,
	\[ \ip{x}{v} = \ip{x}{P^\dagger P v} = \ip{ (P^\dagger)^T x }{P v}. \]
\end{proof}

Finally, we turn our attention to the distribution of $\pl_x^{(\ell)}(\Lambda)$, for a random $x$.
Since we will be studying random linear codes, these will end up being the random variables we would like to get a handle on.
The distribution of $\pl_x^{(\ell)}(\Lambda)$ depends on the linear structure of $\Lambda$, and controlling it is the crux of our argument.

We will rely on a parameter $\sigma_x(\Lambda)$ to measure how ``bad" a set $\Lambda$ is with respect to $x$:
\begin{definition}\label{def:sigma}
For a set $\Lambda \subset \F^d$, we define
\[ \sigma_p(\Lambda) = \EE_{v_1,\ldots, v_p} q^{-\dim( \spn( \inset{ v_1,\ldots, v_p} ))}, \]
where $v_1,\ldots,v_p \in \Lambda$ are drawn uniformly at random with replacement.
\end{definition}

Below in Section~\ref{ssec:sigma}, we prove several useful properties of $\sigma_p(\Lambda)$.

\subsection{Properties of $\sigma_p(\Lambda)$}\label{ssec:sigma}
As we will see, the quantity $\sigma_p(\Lambda)$ is useful because if it is small, then the set $\Lambda$ will be ``well-behaved" in some sense.  On the other hand, if it is large, then $\Lambda$ contains a large, low-dimensional subset, and we may recurse on this subset.  
In this section, we formalize this in a few ways.

First we handle the case for small $p$.
\begin{lemma}\label{lem:easysmallp}
Let $\Lambda \subseteq \F^d$ be of size $L$.  
Choose any integer $T$.
Suppose that for some $p < d$,
\[ \sigma_p(\Lambda) > \frac{1}{q^p} \cdot \inparen{1 + \frac{qT}{L}}^p. \]
Then there is a subset $\Gamma \subseteq \Lambda$ of dimension at most $p$ and size at least $T$.
\end{lemma}
\begin{proof}
Notice that for $p=1$, we have
\[ q \cdot \sigma_1(\Lambda) \leq 1 + \frac{q - 1}{L}, \]
and so the condition that $\sigma_1(\Lambda) > \nicefrac{1}{q} \inparen{ 1 + qT/L }$ is only satisfied for $T = 0$, in which case the conclusion is obviously true.  Thus, the lemma holds if we restrict $p = 1$.
For $p > 1$, let $T$ be the largest size of any subset of $\Lambda$ of dimension at most $p$, and assume inductively that the statement holds for $p - 1$. We write
\begin{align*}
q^p \sigma_p(\Lambda) &= \sum_{s=1}^p \PR{ \dim(v_1,\ldots,v_p) = s } q^{p-s} \\
&\leq \sum_{s=1}^p \inparen{ q^{p-s} \PR{\dim(v_1,\ldots,v_{p-1}) = s-1} + 
q^{p-s} \PR{ \dim(v_1,\ldots,v_{p-1}) = s } \cdot \PR{ v_p \in \spn(v_1,\ldots,v_{p-1}) } }\\
&\leq \sum_{s=1}^p \inparen{ q^{p-s}  \PR{ \dim(v_1,\ldots,v_{p-1}) = s - 1}
+ q^{p-s} \PR{ \dim(v_1,\ldots,v_{p-1}) = s}  \inparen{\frac{T}{L}}  }, 
\end{align*}
using in the last line that the probability that $v_p$ lands in the span of $v_1,\ldots,v_{p-1}$ is at most the probability that $v_p$ lands in the span of the largest $s$-dimensional set, which is at most $T/|\Lambda|$ by assumption.
Re-grouping the terms, we have
\[ q^p \sigma_p(\Lambda) \leq q^{p-1} \sigma_{p-1}(\Lambda) \cdot \inparen{1 + \frac{qT}{|\Lambda|}}. \]
This completes the proof.
\end{proof}

Next, we handle the case for larger $p$.  The issue is that when $p$ is large, we will want our large low-dimensional subset to have dimension significantly less than $p$ (by a constant fraction).
\begin{lemma}\label{lem:sigmaorig}
	Let $\Lambda \subseteq \F^d$.
	Choose $\zeta \in (0,1/4)$, and suppose that $\ell$ satisfies $q \ge \ell^{2/\zeta}$. 
	Suppose that for some $p$ with  $d(1 - \zeta) < p \leq d$, 
\[\sigma_p(\Lambda) > 
\frac{d}{q^{d(1 - 2\zeta)} \cdot \ell^p} \]
	Then there is some set $\Gamma \subseteq \Lambda$ so that $|\Gamma| \geq \zeta|\Lambda|/(qe)$, and so that $\dim(\spn(\Gamma)) \leq d(1 - \zeta)$.
\end{lemma}
\begin{proof}
 	Let $p > 0$ be the $p$ guaranteed in the Lemma's hypothesis. 
	We have (using $p\le d$)
	\begin{align*}
		\frac{p}{q^{d(1 - 2\zeta)} \ell^p} &< \sigma_p(\Lambda) \\ 
		&= \EE_{v_1,\ldots,v_p} q^{-\dim(\spn(v_1,\ldots,v_p))}\\
		&= \sum_{s \leq p} \PR{ \dim(\spn(v_1,\ldots,v_p)) = s } q^{-s},
	\end{align*}
	and so there is some $s \leq p$ so that
	\[ \PR{ \dim(\spn(v_1,\ldots,v_p)) \leq s } > q^{s-d(1 - 2\zeta)}\ell^{-p}. \]
	We observe that this implies that $s < d(1 - \zeta)$.  Indeed, if not, we would have
\begin{align*}
\PR{ \dim( \spn( v_1,\ldots,v_p) ) \leq s } &> q^{d(1 - \zeta) - d(1 - 2\zeta)} \ell^{-d} \\
&\geq q^{\zeta d} \ell^{-d} \\
&\geq 1, 
\end{align*}
using in the first line the assumptions on $p$ and in the final line the assumption that $q > \ell^{1/\zeta}$.  
Since the left-hand-side is a probability, this cannot be true; thus $s < d(1 -\zeta)$. 
Now, we may also bound this probability above by
	\begin{equation}\label{eq:prbound}
	 \PR{ \dim(\spn(v_1,\ldots,v_p)) \leq s } \leq { p \choose s } \inparen{ \frac{T_s}{|\Lambda|} }^{p - s}, 
\end{equation}
where 
\[T_s = \max\inset{ |\Gamma| \suchthat \Gamma \subseteq \Lambda, \dim(\spn(\Gamma)) \leq s }.\]
To see \eqref{eq:prbound}, suppose that $\dim\spn(v_1,\ldots,v_p) \leq s$.  Then there are at most ${p \choose s}$ choices of subsets $I \subset [p]$ so that $\inset{ v_i \suchthat i \in I}$ form a basis for $\spn(v_1,\ldots,v_p)$.  Then for each $i \in [p] \setminus I$, the probability that $v_i$ (which is chosen uniformly at random from $\Lambda$) lands in $\spn(v_1,\ldots,v_p)$ is at most 
\[ \frac{ |\spn(v_1,\ldots,v_p) \cap \Lambda| }{|\Lambda|} \leq \frac{T_s}{|\Lambda|}. \]
These events are independent for $i \in [p]\setminus I$, and a union bound over all the sets $I$ yields \eqref{eq:prbound}.
Thus, we have
\[ q^{s-d(1 - 2\zeta)}\ell^{-p} < \PR{ \dim(\spn(v_1,\ldots,v_p)) \leq s } \leq {p \choose s} \inparen{ \frac{T_s}{|\Lambda|}}^{p-s}, \]
	and hence
\begin{align*}
	 T_s &> |\Lambda| \inparen{ \frac{1}{q}}^{\frac{d(1 - 2\zeta) - s}{p - s}} \cdot \inparen{\frac{1}{\ell}}^{\frac{p}{p-s}} \cdot \inparen{ 1 - \frac{s}{p}} e^{-1} \\
	  &\geq \frac{\zeta|\Lambda|}{e} \inparen{ \frac{1}{q}}^{\frac{d(1 - 2\zeta) - s}{p - s}} \cdot \inparen{\frac{1}{\ell}}^{\frac{p}{p-s}} \qquad \text{using $s < (1 - \zeta)p$} \\
&\ge \frac{\zeta|\Lambda|}{e} \inparen{ \frac{1}{q}}^{ \frac{d(1 - 2\zeta) - s }{p - s}} \cdot \inparen{ \frac{1}{q} }^{\frac{\zeta p }{ 2(p-s)}} \qquad \text{using $q \ge \ell^{2/\zeta}$} \\
&= \frac{\zeta|\Lambda|}{e} \exp_q\inparen{ -\inparen{ \frac{ d(1 - 2\zeta) - s + \zeta p/2 }{p-s}}} \\ 
&\geq \frac{\zeta|\Lambda|}{e} \exp_q\inparen{ -\inparen{ \frac{ p(1 - 2\zeta)/(1 - \zeta) - s + \zeta p/2 }{p-s}}} \qquad \text{using $d(1 - \zeta) < p$}\\ 
	&\geq \frac{\zeta|\Lambda|}{qe}.
\end{align*}
\end{proof}
Together, Lemma~\ref{lem:easysmallp} and \ref{lem:sigmaorig} imply the following:
\begin{lemma}\label{lem:sigma}
	Let $\Lambda \subseteq \F^d$.
	Choose $\zeta > 0$, and suppose that $\ell$ satisfies $q \ge \ell^{2/\zeta}$. 
	Suppose that for some $p > 0$, 
	\[\sigma_p(\Lambda) > \begin{cases} 
\frac{ 1 + 1/q }{q^p } & p \leq (1 - \zeta) d \\
\frac{d}{q^{d(1 - 2\zeta)} \cdot \ell^p} &  (1 - \zeta)d < p \leq d  \end{cases} \]
	Then there is some set $\Gamma \subseteq \Lambda$ so that 
\[|\Gamma| \geq |\Lambda| \cdot \min \inset{ \frac{1}{2dq^2}, \frac{\zeta}{qe} },\] 
and so that $\dim(\spn(\Gamma)) \leq d(1 - \zeta)$.
\end{lemma}
Above, we have used the choice $T = L/(2dq^2)$ in Lemma~\ref{lem:easysmallp}, hence the expression from the conclusion of that lemma is 
\begin{align*}
\inparen{1 + \frac{qT}{L}}^p &\le
\inparen{1 + \frac{1}{2pq}}^p \\
&= 1 + \frac{1}{2q} + \sum_{j \geq 2} {p \choose j} \inparen{\frac{1}{2pq}}^j \\
&\leq 1 + \frac{1}{2q} + \sum_{j \geq 2} \inparen{ \frac{1}{2q} }^j\\
&= 1 + \frac{1}{2q} + \frac{ (2q)^{-2} }{1 - 1/(2q) } \\
&\leq 1 + \frac{1}{q} 
\end{align*}
for $q \geq 2$.
\begin{remark}\label{rem:smallp}
Notice that if we only use $p \leq d(1 - \zeta)$, then the conclusion is that $|\Gamma| \geq |\Lambda|/(2dq^2)$.  This will be useful for our simplified analysis in Section~\ref{sec:easy}.
\end{remark}

\section{Zero-Error list-recovery of random linear codes}\label{sec:easy}
Before we jump in to the proof of the theorem for average-radius list-recovery, 
we illustrate the method by examining the weaker---but still nontrivial---guarantee of list-recovery with no errors.  
Recall from Section~\ref{sec:intro} that we will use the notation $(\ell,L)$-list-recovery to refer to $(1, \ell, L)$-list-recovery---that is, list-recovery with no errors.

\begin{theorem}\label{thm:easy}
Let $\xi > 0$, $\zeta \in (0,1/\nmrk{5})$, choose $\ell \in \mathbb{N}$,  and suppose that
\[q\ge \max  \inset{ \ell^{2/\zeta}, \nmrk{(3\ell)^{\frac{1}{\zeta}-1}}}. \]
Let
\[ R \leq \min \inset{1 - \nmrk{3}\zeta, 1 - \log_q(\ell) - \xi}. \]
Then a random linear code $\cC \subseteq \F_q^n$ of rate $R$ is $(\ell,L)$-list recoverable with probability $1 - o(1)$, where
\[ L > \max\inset{ \frac{2\ell}{\zeta} , \ell \cdot \exp_{2q\nmrk{\cdot \ell/\xi}}\inparen{\frac{2 \log(2\ell/\xi)}{\zeta} }}. \]
\end{theorem}

For the rest of this section, we prove Theorem~\ref{thm:easy}.
As hinted at above, in order to control bad sets of codewords, we will control bad pairs of sets $(X,\Lambda)$.  We say that $(X,\Lambda)$ is \emph{all-bad}\footnote{The terminology ``all-bad" here is in reference to the quantifier $\forall$ in item (d), and is in contrast to the more general case of average-radius list-recovery: in Section~\ref{sec:hard}, we will introduce a slightly different notion of average-badness.} if is forms an obstruction to list-recovery:
\begin{definition}\label{def:allbad}
For a set $X \subseteq \F^d$ with size $n$ and $\Lambda \subseteq \F^d$, we say that $(X,\Lambda)$ is $(L,d,\ell)$-all-bad if
\begin{itemize}
	\item[(a)] $X$ is full-rank,
	\item[(b)] $|\Lambda| \geq L$,
	\item[(c)] $\dim(\spn(\Lambda)) \leq d$,
	\item[(d)] and
$ \forall x \in X, | \inset{ \ip{x}{v} \suchthat v \in \Lambda}| \leq \ell. $
\end{itemize}
We say that a set $X$ is $(L,d,\ell)$-all-bad if there exists a $\Lambda$ so that $(X,\Lambda)$ is $(L, d, \ell)$-all-bad.  For a bad $X$, we will refer to such a $\Lambda$ as a \em witness \em for $X$'s badness.  
\end{definition}

It is easy to characterize $(\ell, L)$-list-recovery in terms of all-bad sets:
\begin{proposition}\label{prop:allbadlr}
	Let $\cC \subseteq \F^n$ be a linear code.  Then $\cC$ is $(\ell, L)$-list-recoverable if and only if, for all $d \leq \dim(\cC)$, there are no sets $X \subseteq \F^d$ so that $\cols(X) \subseteq \cC$ and so that $X$ is $(L, d, \ell)$-all-bad.
\end{proposition}
\begin{proof}
	Let $X = \inset{x_1,\ldots,x_n}$ be a full-rank set. 
	As $\cC$ is linear, the condition that $\cols(X) \subseteq \cC$ is equivalent to the condition that, for all $v \in \F^d$, the vector 
\[ c = ( \ip{x_1}{v}, \ip{x_2}{v}, \ldots, \ip{x_n}{v} ) \]
is in $\cC$.
Thus, a pair of sets $X, \Lambda \subset \F^d$ with $|\Lambda| = L$ describes a set of $L$ codewords in $\cC$.  Conversely, for any set $S$ of codewords in $\cC$, we may find a full-rank $X$ and a set $\Lambda$ of size $L$ so that $(X,\Lambda)$ describes $S$ in this way.

With this correspondence in mind, we see that $\cC$ is $(\ell, L\nmrk{-1})$-list-recoverable if and only if there are no sets $X, \Lambda \subset \F^d$ so that $|X| = n, |\Lambda| = L$, so that $X$ is full-rank and so that for all $i \in [n]$, the set of symbols the occur at position $i$ (which is $\inset{ \ip{x_i}{v} \suchthat v \in \Lambda }$) has size less than or equal to $\ell$.  This is precisely the definition of an all-bad set.
\end{proof}
Thus, we wish to show that with high probability, there are no $(L,d,\ell)$-all-bad sets $X$ with $\cols(X) \subseteq \cC$, for any $d \leq \dim(\cC)$.
We first observe that we may restrict our attention to $(L,d,\ell)$-all-bad sets with $d$ sufficiently small.
Let 
\[ d_0 = \frac{ 2 \ell }{ \xi }, \qquad L_0 = L. \]
\begin{lemma}\label{lem:base2}
Let $\xi > 0$, and suppose that $L > d_0$.
Let $\cC$ be a random linear code of rate $R$ satisfying
\[ R \leq 1 - \log_q(\ell) - \xi. \]
Then for sufficiently large $n$, with probability at least $1 - q^{-\ell n/2}$, if there is an $(L,d,\ell)$-all-bad set $X \subset \F^d$ with $\cols(X) \subseteq \cC$ and with $d > d_0$, then there is an $(L,d_0, \ell)$-all-bad set $X' \subset \F^{d_0}$ with $\cols(X') \subseteq \cC$.
\end{lemma}
\begin{proof}
We will show that, with high probability, for any $(L,d,\ell)$-all-bad set $X$ with $\cols(X) \subseteq \cC$, there is a witness $\Lambda$ for $X$'s badness that has dimension at most $d_0$.  Then Lemma~\ref{lem:projection} will imply the statement.

Fix a set $\Lambda \subset \F^d$ and suppose that $\dim(\Lambda) = d' \geq d_0$.
By Lemma~\ref{lem:projection}, there is a set $\Lambda' \subseteq \F^{d'}$ so that for all $X \in \F^{d}$ so that $(X,\Lambda)$ is $(L,d,\ell)$-all-bad, there is some $X' \in \F^{d'}$ so that $(X', \Lambda')$ is $(L, d', \ell)$-all-bad, and with $\cols(X') \subseteq \cols(X) \subseteq \cC$.

We will count the number of these $X'$, and then use a union bound to control the probability that any of them have $\cols(X) \subseteq \cC$.
Let $v_1,\ldots,v_{d'} \in \F^{d'}$ be linearly independent vectors in $\Lambda'$.  
Imagine choosing $X' \in \F^{n \times d'}$ uniformly at random.  Then
\begin{align*}
\mathbb{P}_{X'} \inset{ (X', \Lambda')   \text{ is  $(L, d', \ell)$-all-bad } } 
&\leq \mathbb{P}_{X'}\inset{ \forall x \in X', \inabs{ \inset{ \ip{x}{v_i} \suchthat i \in [d'] } } \leq \ell \,\wedge\, X' \text{ is full rank } }\\
&\leq \mathbb{P}_{X'}\inset{ \forall x \in X', \inabs{ \inset{ \ip{x}{v_i} \suchthat i \in [d'] } } \leq \ell }\\
&\leq \sum_{S_1,\ldots,S_n \subset \F, |S_i| = \ell} \mathbb{P}_{X'} \inset{ \forall x \in X', \forall i \in [d'], \ip{x}{v_i} \in S_i } \\
&\leq {q \choose \ell}^n \cdot \inparen{ \frac{\ell}{q} }^{d' n}.
\end{align*}
Above, we used
the fact that since the $v_i$ are linearly independent, the random variables $\ip{x}{v_i}$ are independent.
Thus, the number of full-rank $X' \in \F^{d'}$ that are $(L, d', \ell)$-all-bad with witness $\Lambda'$ is at most
\[  q^{d' n} \cdot {q \choose \ell}^n \cdot \inparen{ \frac{ \ell}{q} }^{d'n } \le   q^{\ell n} \cdot \ell^{d'n}. \]

Now, the probability (over a random linear code $\cC$) that any fixed set $\cols(X')$ of $d'$ linearly independent vectors is contained in $\cC$ is at most $q^{-d'n(1 - R)}$, and so for fixed $\Lambda'$ of dimension $d'$, we have
\begin{align*}
&
 \PR{ \exists X , \cols(X) \subseteq \cC \text{ and $X$ is $(L, d', \ell)$-all-bad with witness $\Lambda'$ } }\\
&\qquad \leq q^{\ell n } \cdot \ell^{d' n} \cdot q^{-d' n(1 - R) } \\
&\qquad = \exp_q\inparen{ n \inparen{ \ell - d'(1 - R - \log_q(\ell))}}\\
&\qquad \leq \exp_q \inparen{ n \inparen{ \ell - d'\xi } },
\end{align*}
using the assumption on $R$.
Finally, we union bound over all possible sets $\Lambda'$, and all $d' > d_0$, to find
{\allowdisplaybreaks
\begin{align*}
&\PR{ \exists X \in \F^d, \cols(X) \subseteq \cC, \text{ and $X$ is $(L, d, \ell)$-all-bad with a witness $\Lambda$ of dimension $\geq d_0$}}\\
&\qquad\leq \sum_{d' = d_0}^{\dim(\cC)} \PR{ \exists X \in \F^d, \cols(X) \subseteq \cC, \text{ and $X$ is $(L, d, \ell)$-all-bad with a witness $\Lambda$ of dimension $d'$}}\\
&\qquad\leq \sum_{d' = d_0}^{\dim(\cC)} \PR{ \exists X' \in \F^{d'}, \cols(X) \subseteq \cC, \text{ and $X'$ is $(L, d', \ell)$-all-bad with a witness $\Lambda'$ of dimension $d'$}}\\
&\qquad\leq \sum_{d' = d_0}^{\dim(\cC)} { q^{d'} \choose L } \exp_q\inparen{ n ( \ell - d' \xi ) } \\
&\qquad\leq \sum_{d' = d_0}^{\dim(\cC)} \exp_q\inparen{ n\ell - d'( n \xi - L ) } \\
&\qquad\leq \sum_{d' = d_0}^{\dim(\cC)} \exp_q\inparen{ n\ell - \frac{2\ell}{\xi} \inparen{ n \xi - L } } \qquad \text{by the assumption on $d_0$}\\
&\qquad= \sum_{d' = d_0}^{\dim(\cC)} \exp_q\inparen{ \frac{ 2 \ell L }{\xi} - n\ell  }\\ 
&\qquad\leq q^{-\ell n/2 }
\end{align*}
}
using in the last line the assumption that $n$ is sufficiently large (compared to $L, \xi, \ell$).
\end{proof}

In light of Lemma~\ref{lem:base2}, our goal is now to rule out the possibility that $\cols(X) \subseteq \cC$ for any all-bad $(X,\Lambda)$ with $\dim(\Lambda) \leq d_0$.
Let 
\[ L_j = \frac{L_0}{\nmrk{(2q^2d_0)}^{j}} \qquad \text{and} \qquad  d_j = d_0 \cdot \inparen{ 1 - \zeta }^j \]
for $j = 0, \ldots, j_{\max}$, where
\[ j_{\max} = \frac{\log(d_0)}{\log(1/(1 - \zeta))}. \]
Thus, $d_{j_{\max}} = 1$.
The proof proceeds by induction on $j$, starting with $j_{\max}$ and decrementing $j$. 
We maintain (with high probability) the following inductive hypothesis.
\begin{quote}
\textbf{Hypothesis}$(j)$: There is no $(L_j, d_j, \ell)$-all-bad set $X$ of size $n$ with $\cols(X) \subseteq \cC$.
\end{quote}

\begin{proposition}\label{prop:easybase}
Suppose that 
$ L_0 > \ell \cdot \nmrk{(2q^2d_0)}^{j_{\max}}. $
Then
\textbf{Hypothesis}$(j_{\max})$ holds.
\end{proposition}
\begin{proof}
	If $L_0 \geq \ell \cdot \nmrk{(2q^2d_0)}^{j_{\max}}$, then $L_{j_{\max}} > \ell$.  However, $d_{j_{\max}} = 1$, and any one-dimensional set $\Lambda$ of size $L_{\max}$ has $|\inset{ \ip{x}{v} \suchthat v \in \Lambda }| = |\Lambda| > \ell$ for any nonzero $x$.  Thus, there are no bad sets $X$.
\end{proof}

With the base case out of the way, we prove our main inductive lemma.
\begin{lemma}\label{lem:induct}
Let $j \in \inset{0,\ldots,j_{\max}}$.  Let $\zeta \in (0,1/4)$.  
Suppose that
\[ R < 1 - \nmrk{3}\zeta.\] 
Then with probability at least $1 - q^{- \zeta d_j n}$, at least one of the following holds:
\begin{itemize}
	\item \textbf{Hypothesis}(j) holds; or
	\item \textbf{Hypothesis}(j+1) does not hold.
\end{itemize}
\end{lemma}
Before we prove the lemma, let us note why it suffices to prove the theorem.  We use our assumption on the list size $L = L_0$ to apply Proposition~\ref{prop:easybase}.  This implies that \textbf{Hypothesis}$(j_{\max})$ holds with probability $1$.  Next, noting that the hypothesis of Theorem~\ref{thm:easy} are enough to establish the hypotheses of Lemma~\ref{lem:induct} for all $j$, we proceed inductively.  For every $j$, Lemma~\ref{lem:induct} shows that
\begin{align*}
\PR{ \text{ \textbf{Hypothesis}($j$) does not hold } } &\leq q^{-\zeta n} + \PR{ \text{\textbf{Hypothesis}($j+1$) does not hold}}.
\end{align*}
Since $\PR{ \textbf{Hypothesis}(j_{\max}) \text{ does not hold }} = 0 $, we have
\[ \PR{ \textbf{Hypothesis}(0) \text{ does not hold }} \leq j_{\max} \cdot q^{-\zeta n} \leq q^{-C\zeta n}, \]
for some constant $C$,
using the fact that $n$ is sufficiently large compared to $j_{\max}$ (which does not depend on $n$).
Putting together the above with Lemma~\ref{lem:base2} and Proposition~\ref{prop:allbadlr}, we have
\begin{align*}
&\PR{ \cC \text{ is not $(\ell, L)$-list-recoverable } } \\
&\qquad\leq \PR{ \exists (X,\Lambda) \text{ that are $(L,d,\ell)$-all-bad for some $d \leq d_0$ } } + \PR{ \exists (X, \Lambda) \text{ that are $(L,d,\ell)$-all-bad for some $d > d_0$ } } \\
&\qquad\leq \PR{ \textbf{Hypothesis}(0) \text{ does not hold } } + q^{-\ell n / 2} \\
&\qquad\leq q^{-C \zeta n } + q^{-\ell n / 2 }\\
&\qquad\leq q^{- C' \zeta n },
\end{align*}
using in the last line that $n$ is sufficiently large.  This completes the proof of Theorem~\ref{thm:easy}, modulo the proof of Lemma~\ref{lem:induct}.

\subsection{Proof of Lemma~\ref{lem:induct}}
Finally, we prove Lemma~\ref{lem:induct}.
Say that a set $\Lambda$ is \textbf{good} if the quantities $\sigma_p(\Lambda)$ are small; more precisely, if
\[ \sigma_p(\Lambda) \leq \frac{1 + 1/q}{q^p}. \]
Otherwise, we say that $\Lambda$ is \textbf{bad}.

We will show that most of the sets $X$ that are $(L_j, d_j, \ell)$-all-bad have a witness $\Lambda$ that is bad.  
This means that, with high probability, there won't be a bad $X$ with $\cols(X) \subseteq \cC$ that has a good witness: indeed, there are not very many of these sets, so the probability that they are contained in a random subspace is small.  But this means that if there \em is \em a $(L_j, d_j, \ell)$-all-bad set $X$ with $\cols(X) \subseteq \cC$ (that is, if \textbf{Hypothesis}$(j)$ does not hold), then its witness $\Lambda$ is bad.  We will show that if this is the case, then actually $\Lambda$ contains a large, low-dimensional subset, which can be used to find a violation of \textbf{Hypothesis}$(j+1)$. 

We begin by counting the number of $X$'s that are bad with good witnesses.  
\begin{claim}\label{claim:goodwitness}
	Let $\Lambda \subseteq \F^{d_j}$ be good.   Then the number of $X \subseteq \F^{d_j}$ of size $n$ that are $(L,d_j, \ell)$-all-bad with witness $\Lambda$ is at most $q^{\nmrk{ 2n } d_j \zeta}. $
\end{claim}
\begin{proof}
	Say that $x \in \F^{d_j}$ is \textbf{bad} with respect to $\Lambda$ if
\[ \inabs{ \inset{ \ip{x}{v} \suchthat v \in \Lambda } } \leq \ell. \]
Then because an all-bad set $X$ must have $n$ bad vectors $x$ as elements, we have
\begin{equation}\label{eq:tothen}
\inparen{ \text{number of all-bad sets $X$ with witness $\Lambda$} } \leq \inparen{ \text{number of $x$ that are bad with respect to $\Lambda$} }^n.
\end{equation}
With that in mind, we count the number of bad $x$'s.  Imagine drawing $x \in \F^{d_j}$ uniformly at random.  We have, for all $p \leq (1 - \zeta)d_j$,
\begin{align}
\mathbb{P}_x \inset{ x \text{ is bad for} \ \Lambda } &= \mathbb{P}_x \inset{ \inabs{ \inset{ \ip{x}{v} \suchthat v \in \Lambda } } \leq \ell } \notag\\
&\leq \mathbb{P}_x \inset{ \pl_x(\Lambda) \geq \frac{1}{\ell}} \notag\\
&= \mathbb{P}_x \inset{ \pl_x(\Lambda) - \onebyq \geq \frac{1}{\ell} - \onebyq} \notag\\
&\leq \frac{ \EE_x \inparen{ \pl_x(\Lambda) - \onebyq }^p }{\inparen{ \frac{1}{\ell} - \onebyq}^p } \label{eq:leftoff}.
\end{align}
Now, we will use the assumption that $\Lambda$ is good to bound this expectation. 
If $\Lambda$ is good, then $\sigma_p(\Lambda)$ is small; it turns out that this is sufficient to bound the moments of $\pl_x^{(\ell)}(\Lambda)$:
\begin{claim}
[Follows from Claim \ref{claim:smallp} which is proved later]
\label{claim:smallp1}
Suppose that $\Lambda$ is good. 
Then
for all $p \le d_j(1 - \zeta)$, we have
	\[ \EE_x \inparen{ \pl_x(\Lambda) - \nmrk{\onebyq} }^p \leq \nmrk{ \inparen{ \frac{2}{q} }^p }. \] 
\end{claim}
Using the above in \eqref{eq:leftoff}, and choosing $p = d(1 - \zeta)$, we see
\[ \mathbb{P}_x \inset{ x \text{ is bad for } \Lambda } \leq \inparen{ \frac{\nmrk{2}}{\nmrk{q\cdot}\inparen{\frac{1}{\ell} - \onebyq}} }^{d_j(1 - \zeta)} \leq \inparen{ \frac{\nmrk{3}\cdot \ell}{q}}^{d_j(1 - \zeta)}, \]
where in the last inequality we have used that our lower bound on $q$ implies $(q-\ell)/2\ge q/3$.
Thus, the number of bad $x$'s are bounded by
\[ q^{d_j} \cdot \inparen{ \frac{\nmrk{3}\ell}{q} }^{d_j(1 - \zeta)} = (\nmrk{3} \ell)^{d_j(1 - \zeta)} \cdot q^{d_j \zeta}. \]
Finally, we use the assumption that $q \geq \nmrk{(3\ell)^{1/\zeta-1}}$ 
to conclude that
the number of bad $x$'s are bounded by
$ q^{\nmrk{2d_j} \zeta}. $ 
With \eqref{eq:tothen}, this proves the claim.
\end{proof}

\begin{remark}
\label{rem:q-size}
The above is the only place that the requirement $q \geq \nmrk{(3\ell)^{1/\zeta-1}}$ is used.  This is a convenience, as it simplifies the calculations going forward; moreover, Theorem~\ref{thm:easy} is only interesting when $\ell$ is large, and in this case this requirement is weaker than the requirement that $q \geq \ell^{2/\zeta}$.   However, in Section~\ref{sec:hard}, we will see how to obtain average-radius list-recoverability without this requirement, needing only $q \geq \ell^{2/\zeta}$.  In particular for the list-decoding case, where $\ell = 1$, there will be no constraints on $q$.
\end{remark}

Now, for any fixed set $X \subset \F^{d_j}$, the probability that $\cols(X) \subseteq \cC$ for a random linear code $\cC$ of rate $R$ is
\[ \PR{ \cols(X) \subseteq \cC} \leq q^{-(1 - R)d_jn}. \]
Thus, Claim~\ref{claim:goodwitness} along with the union bound shows that
\[ \PR{ \exists (L_j, d_j, \ell)\text{-all-bad set} X \subseteq \cC } \leq q^{\nmrk{2n}d_j \zeta} q^{-(1 - R)d_j n } = q^{-n d_j(1 - R - \nmrk{2\zeta})} \leq q^{-nd_j \zeta},\]
using the assumption that $R < 1 - \nmrk{3}\zeta$.

This takes care of the good $\Lambda$'s; if \textbf{Hypothesis}$(j)$ fails to hold, then with high probability this is due to a bad $\Lambda$.  However, we have the following claim.
\begin{claim}\label{claim:easyversion}
	Suppose that $\Lambda \subset \F^{d_j}$ is bad, and that $X \subset \F^{d_j}$ is a $(L,d_j,\ell)$-all-bad set with witness $\Lambda$.  Then \textbf{Hypothesis}$(j+1)$ fails to hold.
\end{claim}
\begin{proof}
	If $\Lambda \subset \F^{d_j}$ is bad, then by Lemma~\ref{lem:sigma} (with Remark~\ref{rem:smallp}) there is some set $\Gamma \subseteq \Lambda$ of size at least $L_j/(\nmrk{2q^2d_0}) = L_{j+1}$ with dimension at most $d_j(1-\zeta)=d_{j+1}$.
Further, $\Gamma$ is still a witness for $X$'s badness, since for all $x \in X$,
\[ |\inset{ \ip{x}{v} \suchthat v \in \Gamma} | \leq | \inset{ \ip{x}{v} \suchthat v \in \Lambda } | \leq \ell. \]
By Lemma~\ref{lem:projection}, there is some $X'$ of size $n$ that is $(L_{j+1}, d_{j+1}, \ell)$-all-bad.  Thus, \textbf{Hypothesis}$(j+1)$ does not hold.
\end{proof}
This completes the proof of Lemma~\ref{lem:induct}, and hence Theorem~\ref{thm:easy}.

\section{Average-radius list-recovery of random linear codes}\label{sec:hard}
Now, we adapt our argument from Section~\ref{sec:easy} to establish \em average-radius \em list-recoverability.
Our main theorem is as follows.

\begin{theorem}\label{thm:avgrad}
There are constants $C, C'$ so that the following holds.
Choose $\ell \in \mathbb{N}$, $\ell > 0$.
\nmrk{Let $\barmu\in\inset{0,\centermu}$ and let $\beta = \arm{(q+1)}^{\zeta/(2(1-\zeta))}$ if $\barmu=0$ and $\beta=2$ otherwise.}
Choose \nmrk{$\eps > \beta\cdot\centermu+\barmu$}, and $\eta, \zeta, \xi \in (0,1)$ so that 
\begin{equation}
\label{eq:eps-condn-gen}
 (1 - \zeta)(\eps - \eta) > \nmrk{\beta\cdot\centermu+\barmu}, \qquad \zeta \leq 1/20. 
\end{equation}
Suppose that
\[ q \geq  \max \inset{ 2, \ell^{2/\zeta}} \]
and
\[ R \leq \min \inset{ \inparen{\eps - \nmrk{\beta\cdot\centermu-\barmu}}(1 - 5\zeta) - \eta, \arm{1 - H_{q/\ell}(1 - \eps + \eta) - \log_q(\ell)} - \xi}. \]
Let $\cC$ be a random linear code over $\F = \F_q$ of rate $R$.
Then for sufficiently large $n$, with probability at least $1 - \exp(-C \zeta R n) - \exp(- C \xi n)$,  $\cC$ is $(\mrk{\eps\ml}, L)$-average-radius list-\mrk{recoverable},
for
\[ 
L \geq \inparen{ \frac{ 1 - \eps + \eta }{\eta} } \cdot \exp_{\nmrk{q\cdot \frac{\ell}{\xi}}}\inparen{ \frac{ C' \log(\ell \zeta/\xi ) }{\zeta } } \cdot \inparen{\frac{1}{ \eps - \eta } }^{C' \log^2(\ell \zeta/\xi )/\zeta^3 }.
\]
\end{theorem}
\vspace{.3cm}

Theorem~\ref{thm:avgrad} is a bit complicated to look at.  
In Section~\ref{sec:consequences} below, we will simplify it and we will show how to derive the results stated in Section~\ref{sec:results} from Theorem~\ref{thm:avgrad}.  Then we will continue with the proof of the theorem in Section~\ref{sec:avgproof}.

\subsection{Consequences of Theorem~\ref{thm:avgrad}}\label{sec:consequences}
\label{ssec:cors}

In this section, we show how a few applications of Theorem~\ref{thm:avgrad}, which will imply the results advertised in Section~\ref{sec:results}. 
For the reader's convenience, we will also re-state the corollaries here.

We begin with a simplification of Theorem~\ref{thm:avgrad} that we will use for all of our corollaries.  We stated Theorem~\ref{thm:avgrad} as we did because it fits better with the proof, but it is more useful in the following form: 
\begin{corollary}\label{cor:avgrad}
There are constants $C, C'$ so that the following holds.
Choose $\ell \in \mathbb{Z}$, $\ell > 0$.
Choose $0 < \zeta, \eta, \xi < 1/20$, and choose $\eps$ so that
\begin{equation}
\label{eq:eps-condn-cor}
\nmrk{
 \eps > \eta+ \frac{\ell(1+2\zeta)}{q} \inparen{ 1 + \min\inset{2,\frac{ \zeta \log(\arm{q+1}) }{1 - \zeta} }}.
}
\end{equation}
Suppose that
$ q \geq  \max \inset{ 2, \ell^{2/\zeta}}$.
Let
\[ R_0 := \inparen{ \eps - \frac{\ell}{q} \inparen{1 + \min\inset{2,\frac{ \zeta \log(\arm{q+1}) }{1 - \zeta} }}}(1 - 5\zeta). \]
Suppose that
\[ R \leq \min \inset{ R_0 - \eta,  \arm{1 - H_{q/\ell}(1 - \eps + \eta)- \log_q(\ell)} - \xi}. \]
Let $\cC$ be a random linear code over $\F = \F_q$ of rate $R$.
Then for sufficiently large $n$, with probability at least $1 - \exp(-C \eps n) - \exp(- C \xi n)$,  $\cC$ is $(\mrk{\eps\ml}, L)$-average-radius list-\mrk{recoverable},
for some
\[ 
L \leq \inparen{ \frac{ 1 - \eps + \eta }{\eta} } \cdot \exp_{\nmrk{q\cdot \frac{\ell}{\xi}}}\inparen{ \frac{ C' \log(\ell \zeta/\xi ) }{\zeta } } \cdot \inparen{\frac{1}{ \eps - \eta } }^{C' \log^2(\ell \zeta/\xi )/\zeta^3 }.
\]
Further, if $q\ge \frac{2\ell}{(1-\zeta)\eps}\cdot \log\inparen{\frac{2\ell}{(1-\zeta)\eps}}$, then $R_0$ may be taken to be
\[ R_0 = \inparen{\eps - \centermu}(1 - 6\zeta). \]
\end{corollary}

\begin{proof}
We first note that since $\frac{1}{1-\zeta}\le 1+2\zeta$, the condition on $\eps$ in~\eqref{eq:eps-condn-cor} implies the condition in~\eqref{eq:eps-condn-gen}.
Next, we show that the bound of
\[R \le \inparen{\eps - \nmrk{\beta\cdot\centermu-\barmu}}(1 - 5\zeta)\]
in Theorem~\ref{thm:avgrad}
can be replaced by
\begin{equation}
\label{eq:simpler-R-bound-gen}
 R \le \inparen{\eps - \centermu\cdot\inparen{1+\min\inset{2,\frac{\zeta\log{\arm{(q+1)}}}{1-\zeta}}}}(1 - 5\zeta).
\end{equation}
Further, if $q\ge \frac{\ell}{(1-\zeta)\eps}\cdot \log\inparen{\frac{\ell}{(1-\zeta)\eps}}$, then the bound on $R$ can be simplified to
\begin{equation}
\label{eq:simpler-R-bound-large-q}
R \le \inparen{\eps - \centermu}(1 - 6\zeta).
\end{equation}

To see these claims, by the definition of $\barmu$ and $\beta$ we have that
\[\beta\cdot \centermu + \barmu \le \min\inset{3,\exp\inparen{\frac{\zeta\log{\arm{(q+1)}}}{2(1-\zeta)}}}\cdot\centermu.\]
Then the claimed bound in~\eqref{eq:simpler-R-bound-gen} follows from the fact that for all $x>0$ we have $\min(3,e^x) \le 1+\min(2x,2)$. Further, note that~\eqref{eq:simpler-R-bound-gen} implies~\eqref{eq:simpler-R-bound-large-q} if
\[\inparen{\eps - \centermu\cdot\inparen{1+\frac{\zeta\log{\arm{(q+1)}}}{1-\zeta}}}(1 - 5\zeta) \ge \inparen{\eps - \centermu}(1 - 6\zeta),\]
which is implied by
\[\eps\zeta \ge \centermu\cdot\inparen{\frac{\zeta\log{\arm{(q+1)}}}{1-\zeta}}(1 - 5\zeta) +\frac{\zeta\ell}{q},\]
which in turn is implied by the lower bound on $q$.
\end{proof}

Corollary~\ref{cor:avgrad} is perhaps still a bit hard to parse.  Below, we instantiate it in two different ways to prove the corollaries advertised in Section~\ref{sec:results}.  To give additional intuition about the rate expression that appears in Corollary~\ref{cor:avgrad}, in Figure~\ref{fig:rates}, we plot the rate $R$ for several values of $q$ and for $\ell = 1$.    
\begin{figure}[h]
	\begin{center}
		\includegraphics[width=4in]{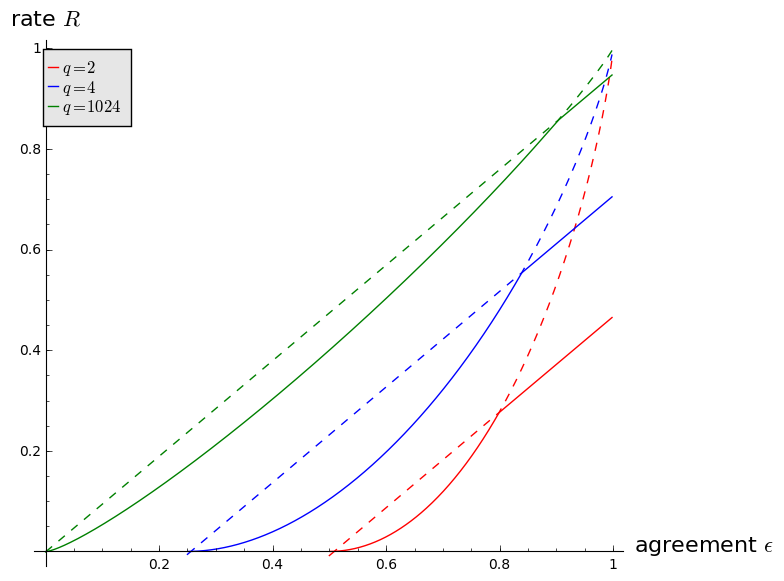}
	\end{center}
	\caption{The rate/agreement trade-off guaranteed in Corollary~\ref{cor:avgrad}.  We consider $q = 2,4,1024$, with $\eta, \xi = 0$, and $\zeta = 0.01, \ell = 1$, and plot both terms in the rate in Corollary~\ref{cor:avgrad}; the minimum is shown as a solid line.  When the curved line is below the straight line, the rate guarantee of Corollary~\ref{cor:avgrad} can approach the optimal rate of $1 - H_q(\eps)$.}
	\label{fig:rates}
\end{figure}

From Corollary~\ref{cor:avgrad}, we can now establish the results advertised in Section~\ref{sec:results}.  For the reader's convenience, we restate these corollaries here.  The following expressions on the entropy $H_q(x)$ will be useful in our derivations below.
First, we will use the fact that the Taylor expansion of $H_q(1 - 1/q - x)$ about $x=0$ is
\begin{equation}\label{eq:entropyexp}
H_q(1 - 1/q - x) = 1 + \sum_{j=1}^\infty \frac{(-1)^j}{j(j+1)}\inparen{ 1 - \inparen{ \frac{-1}{q-1} }^j }  \inparen{ \frac{ q^j }{\ln(q)}}  x^{j+1}.
\end{equation}
We will also use that the Taylor expansion of $H_q(y)$ about $q = \infty$ is
\begin{equation}\label{eq:entropyinf}
H_q(y) = y\log_q(q-1) + \frac{H_2(y)}{\log_2(q)} = y + \frac{H_2(y)}{\log_2(q)} - \frac{ y}{\ln(q)} \sum_{j=1}^\infty \frac{ q^{-j}}{j}. 
\end{equation}

We first address list-recovery and list-decoding when $q$ is large, and the agreement fraction $\eps$ is small, close to $\ell/q$.

\begin{corollary}[List-decoding and list-recovery over large alphabets (Corollary~\ref{cor:largeq} restated)]
Let $\ell\ge 1$ be an integer.
For every sufficiently small constant $\gamma > 0$, there are constants $C, C'$ (which depend on $\gamma$) so that the following holds.
Choose $\delta > 0$ sufficiently small.
Let $q \geq \max \inset{C(\ell/\delta)^2, \ell^{C/\arm{\delta}}}$.
Suppose that
\[ R \leq \inparen{1 - H_{q/\ell}\inparen{1 - \frac{\ell}{q} - \delta} - \log_q(\ell)}(1 - \gamma). \]
Then a random linear code of rate $R$ over $\F_q$ is $(\eps, \ell, L)$-average-radius list-recoverable with high probability, for
\[ \eps = \frac{\ell}{q} + \delta \]
and 
\[ L  \leq q^{ C' \log^2(\ell/\delta) } . \]
\end{corollary}
\begin{proof}
We instantiate Corollary~\ref{cor:avgrad} with $\eta, \xi = \zeta \delta$ and with $\zeta = C''\gamma$ for a constant $C''<1$ to be chosen later.

We first verify the these parameters satisfy~\eqref{eq:eps-condn-cor} and claimed bound on $L$. Towards this, note that to satisfy~\eqref{eq:eps-condn-cor}, we need
\[ \delta > \eta +\centermu\inparen{2\zeta+ (1+2\zeta)\inparen{1+\frac{\zeta\log{\arm{(q+1)}}}{1-\zeta} }}.\]
By our choice of parameters, the above would follow as long as
\[\centermu\inparen{2\zeta+ (1+2\zeta)\inparen{1+\frac{\zeta\log{\arm{(q+1)}}}{1-\zeta} }} < (1-\zeta)\delta,\]
which in turn is true if $q\ge \Omega_{\gamma}(\ell/\delta\log(\ell/\delta))$, which in turn is implied by our lower bound of $q\ge C\cdot\inparen{\frac{\ell}{\delta}}^{2}$. The claimed bound on $L$ follows from the fact that with our choices of parameters, $\frac{\zeta\ell}{\xi}=\frac{\ell}{\delta}$. Further, $\frac{1}{\eps-\eta}\le q$ and $\frac{q\ell}{\delta}\le \poly(q)$.

Next, we consider the rate.  There are two terms in Corollary~\ref{cor:avgrad}.  For the first $R_0 - \eta$ term, we note that $q\ge C\cdot \inparen{\nmrk{\frac{\ell}{\delta}}}^{2}$, we have $q\ge \frac{2\ell}{(1-\zeta)\eps}\cdot \log\inparen{\frac{2\ell}{(1-\zeta)\eps}}$ and hence,
 in this parameter regime, $R_0 - \eta$ is bounded by
\begin{align*}
 \inparen{ \eps - \frac{\ell}{q}  }(1 - 6\zeta) - \eta
&= \delta(1 - O(\zeta)) - \eta  \\
&\geq \delta(1 - \gamma).
\end{align*}

We now study the second term in the rate expression in Corollary~\ref{cor:avgrad}, which is
\begin{align*}
\arm{1 - H_{q/\ell}(1 - \eps + \eta) - \log_q(\ell)} - \xi 
&\geq (1 - H_{q/\ell}(1 - \eps) - C'''\eta) - \log_q(\ell)  - \xi \\
&\geq \arm{1 - H_{q/\ell}(1 - \eps) - \log_q(\ell) } - (C''' + 1)C''\delta \gamma,
\end{align*}
where the constant $C'''$ is a bound on the derivative of $H_{q/\ell}(x)$ at $1 - \eps$: this is bounded because of the assumption that $\eps$ is sufficiently small.  The bound $C'''\eta \leq C'''C''\delta \gamma$ follows from the choice of $\xi, \eta = \zeta \delta$ and $\zeta = C''\gamma$.  

Next, we note that
\begin{equation}
\label{eq:1-Hq-lb-ub-orig}
 C''''\delta \leq 1 - H_{q/\ell}(1 - \eps) \leq \delta,
\end{equation}
where $C''''<1$ is some absolute constant.   
Indeed using~\eqref{eq:entropyinf}, we see that
\[\inparen{1-\centermu -\delta}\inparen{1-\frac{\ell}{(q-\ell)\ln(q/\ell)}} \le H_{q/\ell}\inparen{1-\centermu-\delta}-\frac{H_2\inparen{\centermu+\delta}}{\log_2(q/\ell)} \le \inparen{1-\centermu -\delta}.\]
The above implies
\begin{equation}
\label{eq:1-Hq-lb-ub}
\centermu +\delta \le 1-H_{q/\ell}\inparen{1-\centermu+\delta}+\frac{H_2\inparen{\centermu+\delta}}{\log_2(q/\ell)} \le \centermu +\delta +\frac{\ell}{(q-\ell)\ln(q/\ell)}.
\end{equation}
We first argue the upper bound in~\eqref{eq:1-Hq-lb-ub-orig}. Using the above along with the inequality $H_2(x)\ge x\log_2(1/x)$, we get
\begin{align*}
1-H_{q/\ell}\inparen{1-\centermu+\delta} &\le \centermu +\delta +\frac{\ell}{(q-\ell)\ln(q/\ell)} - \inparen{\centermu+\delta}\frac{\log_2\inparen{\frac{q}{\ell+\delta q}}}{\log_2(q/\ell)}\\
&\le \delta+2\cdot \centermu - \frac{\delta}{\log_2(q/\ell)}\\
&\le \delta,
\end{align*}
as desired. In the second inequality we use the fact that for large enough $q$, we have $(q-\ell)\ln(q/\ell)\ge q$ and $q\ge 2(\ell+\delta q)$ while the last inequality uses the fact that our lower bound on $q$ implies that $\frac{q/\ell}{\log_2(q/\ell)} >2/\delta$.

To argue the lower bound in~\eqref{eq:1-Hq-lb-ub-orig}, the lower bound in~\eqref{eq:1-Hq-lb-ub} along with the inequality $H_2(x) \le x(\log_2(1/x)+\log_2{e})$, implies that
\begin{align*}
1-H_{q/\ell}\inparen{1-\centermu+\delta} &\ge \inparen{ \centermu+\delta} \inparen{ 1 - \frac{\log_2\inparen{ \frac{q}{\ell+\delta q}}+\log_2{e}}{\log_2(q/\ell)}}\\
&\ge \delta \inparen{ 1 - \frac{\log_2\inparen{ \frac{q}{\ell(1+\sqrt{C q})}}+\log_2{e}}{\log_2(q/\ell)}}\\
&\ge C''''\delta,
\end{align*}
for some $C''''>0$. In the above, the second inequality follows from our assumption that $\delta\ge \ell\sqrt{\frac{C}{q}}$ and the final inequality follows from noting that for large enough $q$ the expression in the second parenthesis in the second line converges to  $1/2$.
This establishes \eqref{eq:1-Hq-lb-ub-orig}.

Thus, the second term in the minimum for the rate bound is bounded by
\begin{align*}
1 - H_{q/\ell}(1 - \eps - \eta) - \log_q(\ell) - \xi 
&\geq 1 - H_{q/\ell}(1 - \eps) - \log_q(\ell) - (C'''+1)C'' \delta \gamma \\
&\geq \inparen{ 1 - H_{q/\ell}(1 - \eps) - \log_q(\ell) }(1 - \gamma).
\end{align*}
Indeed, the last inequality requires that
\[ 1 - H_{q/\ell}(1 - \eps) - \log_q(\ell) > (C''' + 1)C'' \delta,\]
which follows by choosing $C''$ appropriately small in terms of $C''', C''''$, and using $\log_q(\ell) \leq \frac{\delta}{C}$ and choosing $C$ appropriately large in terms of $C''', C''''$.

Using again the fact that $1 - H_{q/\ell}(1 - \ell/q - \delta) \leq \delta$ (which in turn implies that $(1 - H_{q/\ell}(1 - \eps) - \log_q(\ell))(1 -\gamma)  < \delta(1-\gamma)$), we conclude that the second term in the minimum is the smaller one, and that the following upper bound on $R$
\[ R \leq (1 - H_{q/\ell}(1 - \eps) - \log_q(\ell))(1 - \gamma),\]
implies the bound on $R$ in Corollary~\ref{cor:avgrad}, which now implies the result.

\end{proof}

Next we consider the parameter regime where $q$ is constant and $\eps$ is bounded away from $1/q$.
\begin{corollary}[Average-radius list-decoding over constant-sized alphabets; Corollary~\ref{cor:constantagr} restated]
\label{cor:constantagr2}
Choose any prime power $q \geq 2$.  There are constants $C, \delta_0$, and $\eps_0 \lneq \eps_1$ (which depend on $q$) so that the following is true.
For all $\delta \in (0, \delta_0)$, and for all $\eps \in (\eps_0, \eps_1)$, a random linear code over $\F_q$ of rate
\[ R = 1 - H_q(1-\eps) - \delta \]
is $(\eps, L)$-average-radius list-decodable with list size
\[ L \leq \inparen{ \frac{C}{\delta} }^{C \log^2(1/\delta)}. \]

Moreover, for $q = 2$, we may take $\eps_0 = 0.51$ and $\eps_1 = 0.8$, and in general we may take
\[ \eps_0 = 1/q + 1/q^2 , \qquad  \eps_1 = \max \inset{ 0.8,  1 - \inparen{ \frac{ 1.1 \cdot \ln(\arm{q+1}) }{q} }}. \]
\end{corollary}
\begin{proof}
Fix a constant $q$ and $\delta > 0$.
In Corollary~\ref{cor:avgrad}, we choose $\xi, \eta = \Theta(\delta)$.
\nmrk{
We first argue that with these parameter choices, we satisfy~\eqref{eq:eps-condn-cor} and claimed bound on $L$. Note that~\eqref{eq:eps-condn-cor} is satisfied if
\begin{align*}
\frac{1}{q}+\frac{1}{q^2}&> \eta + \frac{(1+2\zeta)}{q}\inparen{1+\frac{\zeta\log{\arm{(q+1)}}}{1-\zeta}}\\
&= \frac{1}{q}+\frac{1}{q}\inparen{ q\eta+2\zeta+\frac{(1+2\zeta)\zeta\log{\arm{(q+1)}}}{1-\zeta} }.
\end{align*}
The above is satisfied if $\inparen{ q\eta+2\zeta+\frac{(1+2\zeta)\zeta\log{\arm{(q+1)}}}{1-\zeta} }<\frac{1}{q}$, which is satisfied if $\delta$ and $\zeta$ are taken to be small enough, in terms of $q$.  (More precisely, we should have $\delta\le O(1/q^2)$ and $\zeta=O\inparen{\frac{1}{q\log{q}}}$).
For the bound on $L$ notice that $\zeta$ depends only on $q$, while $\xi$ and $\eta$ are on the order of $\delta$.  Thus, replacing everything that depends only on $q$ in the list size bound by a constant $C'''$, we obtain the bound
\begin{align*}
L &\leq \inparen{\frac{1 - \eps + \eta}{\eta} }\cdot \exp_{q /\xi}\inparen{\frac{C' \log( \zeta /\xi)}{\zeta}} \cdot \inparen{\frac{1}{ \eps - \eta}}^{C' \log^2( \zeta/\xi)/\zeta^3} \\
	&\le \inparen{\frac{C'''}{\delta}} \exp_{C'''/\delta}\inparen{C''' \log(1/\delta)} \cdot \exp_{C'''} \inparen{C''' \log^2( 1 / \delta)} \\
	&\leq \inparen{\frac{C}{\delta}}^{C \log^2(1/\delta))},
\end{align*}
for some constant $C$ depending only on $q$, as desired.
}

Then the minimum in the rate expression in Corollary~\ref{cor:avgrad} is $\min\inset{R_0, R_1}$, where
\[ R_0 := \inparen{ \eps - \frac{1}{q} \inparen{1 + \frac{ \zeta \log(\arm{q+1}) }{1 - \zeta} }}(1 - 5\zeta), \]
and 
\[ R_1 := 1 - H_q( 1 - \eps + \eta) - \xi = 1 - H_q(1 - \eps) - \delta, \]
by choosing $\eta, \xi = \Theta(\delta)$ appropriately.
 The only step left is to show that for $\eps \in (\eps_0,\eps_1)$, for some $\eps_0 \lneq \eps_1$, we have $R_1 < R_0 - \delta$.

As $R_0$ is a linear function of $\eps$ and $R_1$ is convex, there may be two intersection points for these curves, which will be $\eps_0$ and $\eps_1$, respectively.  
By choosing $\zeta = 0.01$, we may establish the claim for $q = 2$ numerically.  (See Figure~\ref{fig:rates}).
For general $q$,
we begin with $\eps_0$, which we claim may be taken to be $\eps_0 = 1/q + 1/q^2$.  To see this, we use \eqref{eq:entropyexp} to see that
with this choice of $\eps_0$, we have 
\begin{align*}
R_1 +\delta &= 1 - H_q( 1 - \eps_0 ) \\
&= 1 - H_q( 1 - 1/q - 1/q^2 ) \\
&= -\sum_{j=1}^\infty \frac{ (-1)^j }{j(j+1)} \inparen{ 1 - \inparen{ \frac{-1}{q-1} }^j } q^j q^{-2(j + 1)} / \ln(q)\\
&= -\sum_{j=1}^\infty \frac{ (-1)^j }{j(j+1)} \inparen{ 1 - \inparen{ \frac{-1}{q-1} }^j } q^{-2 - j} /\ln(q) \\
&\leq \frac{1}{q^2} \sum_{j=1}^\infty \frac{ 1}{q^j \ln(q)\cdot j \cdot (j+1) } \\
&\leq \frac{1}{2 \cdot q^2}. 
\end{align*}
On the other hand, we have $R_0$ given by
\begin{align*}
R_0 &= \inparen{ \eps_0 - \frac{1}{q} \inparen{1 + \frac{ \zeta \log(\arm{q+1}) }{1 - \zeta} }}(1 - 5\zeta) \\
&= \inparen{ \frac{1}{q^2} - \frac{ \zeta \log(\arm{q+1})}{(1-\zeta)q} }\inparen{1 - 5\zeta}.
\end{align*}
By choosing 
\[ \zeta \leq \min \inset{ \frac{ 1 }{4 q \log(\arm{q+1}) }, 1/50 }, \]
we can thus ensure that 
\[ R_1 < R_0 - \frac{7}{40 \cdot q^2} < R_0 - \delta, \]
by choosing $\delta$ to be sufficiently small compared to $1/q^2$.
 
Next we consider choosing $\eps_1$.
Again, to see this, it suffices to show that for $\eps = \eps_1$, we have $R_1 < R_0 - \delta$.
This will be true (for some $\zeta, \delta$, which may depend on $q$) as long as
\[ 1 - H_q(1 - \eps_1) \lneq \eps_1 - \frac{1}{q}. \]
Numerically, we see that this is satisfied for all $q \geq 2$ provided that $\eps_1 \geq 0.8$. 
To understand how we may take $\eps_1$ as $q$ grows, we use
use \eqref{eq:entropyinf} to approximate
\begin{align*}
R_1 &= 1 - H_q( 1 - \eps_1 ) \\
& = 1 - \inparen{ 1 - \eps_1 + \frac{H_2(1 - \eps_1)}{\log_2(q)} - \frac{ 1 - \eps_1 }{\ln(q)} \sum_{j=1}^\infty \frac{1}{j \cdot q^j } }\\
&\leq \eps_1 - \frac{H_2(1 -\eps_1) }{\log_2(q)} + \frac{1 - \eps_1}{(q-1) \ln(q) } \\
&\leq \eps_1 - \frac{1 - \eps_1}{\ln(q)} \inparen{ \ln(1/(1 - \eps_1)) - \frac{1}{q-1 } },
\end{align*}
using the approximation that $H_2(x) \geq x \log(1/x)$.
On the other hand, we have
\begin{align*}
R_0 &= \inparen{ \eps_1 - \frac{1}{q} \inparen{1 + \frac{ \zeta \log(\arm{q+1})}{1 - \zeta} } }(1 - 5\zeta) 
\end{align*}
and so the condition that $R_1 < R_0 - \delta$ can be met (for some $\zeta, \delta$ which depend on $q$) provided that
\[ \frac{1}{q} \lneq \frac{1 - \eps_1}{\ln(q)} \inparen{ \ln(1/(1 - \eps_1)) - \frac{1}{q-1 } }. \]
Choose 
\[ \eps_1 = 1 - 1.1 \cdot \ln(\arm{q+1}) /q, \]
so that the above reads
\[ \frac{1}{q} \lneq \frac{1.1}{q} \inparen{ \ln \inparen{ \frac{q}{ 1.1\ln(\arm{q+1}) } } - \frac{1}{q-1} }, \]
which holds for $q \geq 6$.
Finally, we observe that for $q \leq 6$, this value of $\eps_1$ is in fact smaller than $0.8$, and so the claim holds.
\end{proof}

Finally, we consider high-rate list-recovery.
\begin{corollary}[High-rate list-recovery (Corollary~\ref{cor:highratelr} re-stated)]
	There are constants $C, \gamma_0$ so that the following holds.
Choose $0 < \gamma < \gamma_0$ sufficiently small and let $\ell > 1$ be an integer.
Suppose that $q \geq  \ell^{C/\gamma}$, and let $\cC$ be a random linear code of rate $R$ for
\[ R \leq 1 - \gamma. \]
Then with high probability, $\cC$ is $(1 - \gamma/10, \ell, L)$-list-recoverable, for some
\[ L \leq \inparen{ \frac{q\ell}{\gamma}}^{\log(\ell)/\gamma} \cdot \exp\inparen{\frac{\log^2(\ell)}{\gamma^3}}. \]
\end{corollary}
\begin{proof}
	In Corollary~\ref{cor:avgrad}, we will choose $\zeta, \eta, \xi = \gamma/20$, and $\eps = 1 - \gamma/10$.   Let $C' = 1/10$, so that $\eps = 1 - C'\gamma$.  
	
We first notice that the by choosing $C$ large enough, we can guarantee that the condition $q \geq \ell^{2/\zeta}$ is satisfied.	
	Next, we observe that the condition~\eqref{eq:eps-condn-cor} amounts to requiring $1 - \frac{3}{20}\cdot\gamma \geq 4\ell / q$, which is satisfied by choosing $\gamma \leq \gamma_0$ sufficiently small.  Next, we consider the terms in the rate.  The first is
	\begin{align*}
		R_0 - \eta &= \inparen{ \eps - \frac{3\ell}{q}}(1 - 5\zeta)  - \eta\\
		&= \inparen{ 1 - C'\gamma - \frac{3\ell}{q} }(1 - 5\zeta) - \eta. 
		\end{align*}
	The assumption that $q \geq \ell^{C/\gamma}$ (and the fact that $\ell > 1$) implies that $3\ell/q \leq C'\gamma$, by choosing $C$ large enough.  Finally, by choosing $\zeta, \eta = \gamma/20$, we have
	\[ R_0 - \eta \geq (1 - 2C'\gamma)(1 - \gamma/4) - \frac{\gamma}{20}
	\geq 1 - \inparen{ 2C' + 1/4 + 1/20 }\gamma \geq 1 - \gamma,
	\]
	using the choice $C' \leq 1/10$.
	The second term is
	\begin{equation}\label{eq:secondterm}
		\arm{1 - H_{q/\ell}(1 - \eps + \eta) - \log_q(\ell)} - \xi 
		= \arm{ 1 - H_{q/\ell}(C'\gamma + \gamma/20) - \gamma/C} - \gamma/20.
	\end{equation}
We estimate by~\eqref{eq:entropyinf}
\begin{align*}
	H_{q/\ell}(C' \gamma + \gamma/20) &\leq (C' + 1/20)\gamma + \frac{H_2(C' + 1/20)\gamma)}{\log_2(q/\ell)} \\
	&\leq \inparen{ C' + \frac{1}{20}}\gamma + \frac{2(C' + 1/20) \gamma \log( 1/\gamma)}{\log_2(q/\ell)} \\
	&\leq 3 \cdot \inparen{ C' + \frac{1}{20}}\gamma \\
	&\leq \frac{\gamma}{2},
	\end{align*}
using in the second-to-last line that $q/\ell \geq 1/\gamma$, which is again guaranteed by choosing $C$ sufficiently large.  
Thus, returning to \eqref{eq:secondterm}, we see that this is bounded below by
\[ \arm{ 1 - \gamma/2 - \gamma/C} - \gamma/20 \geq 1- \gamma,\]
again by ensuring $C$ is sufficiently large.  Thus, the rate expression in Corollary~\ref{cor:avgrad} may be replaced by $1 - \gamma$.
Plugging our choices of parameters into the expression for the list size completes the proof.
\end{proof}

\subsection{Proof of Theorem~\ref{thm:avgrad}}\label{sec:avgproof}
In the remainder of the Section, we prove Theorem~\ref{thm:avgrad}.

The main idea behind Theorem~\ref{thm:avgrad} will be to study \em average-badness \em instead of all-badness.  We will discuss this more below, but briefly, we will say a set $(X,\Lambda)$ is average-bad if $\sum_{x \in X} \pl_x^{(\ell)}(\Lambda)$ is large.
The argument proceeds as in Section~\ref{sec:easy}: in the first step, we show that we only need to consider small-dimensional $\Lambda$, by exploiting the linear independence of large sets $\Lambda$.  Then the argument proceeds by induction, following the same outline as in the proof of Theorem~\ref{thm:easy}.
However, there are a few complications that prevent the argument from porting over directly.  Below, we briefly mention three of these and explain our techniques to overcome them. 
\begin{enumerate}
	\item In Lemma~\ref{lem:base2}, the argument was standard; for average-radius list-recoverability, we will have to modify the argument slightly to work for average-radius list-recoverability, taking a slight loss in the list size.  This part of the argument follows ideas from an earlier argument in~\cite{RW15} that list-decodability implies average-radius list-decodability, with some loss in parameters.
	\item In~\eqref{eq:tothen}, we bounded the number of all-bad $X$ by the number of bad $x$ raised to the $n$.  This works for all-badness, but not for average-badness.  Instead, we will need to use a Chernoff-like analysis to bound the number of average-bad $X$.  This will require more careful control of the moments of $\pl_x^{(\ell)}(\Lambda)$.
	\item In Claim~\ref{claim:easyversion}, we used the fact that if $\Lambda$ is a witness for the all-badness of $X$, then any subset of $\Lambda$ is also a witness for that all-badness.  Indeed, the size of the set $|\inset{ \ip{v}{x} \suchthat v \in \Lambda, x \in X}|$ can only shrink when we pass to a subset of $\Lambda$.  However, for average-badness, this is no longer true: if $\pl_x^{(\ell)}(\Lambda)$ is large, then there might be subsets $\Lambda' \subseteq \Lambda$ so that $\pl_x^{(\ell)}(\Lambda')$ is small.  To get around this, we will have to introduce another step in the argument, which shows that if $X$ is average-bad with a bad witness $\Lambda$, then not only can we find a large low-rank subset of a bad $\Lambda$, but moreover we can find on that is still a witness for $X$'s badness.
\end{enumerate}

With this in mind, we proceed with the proof.  Paralleling the proof of Theorem~\ref{thm:easy}, we say that a set $X$ is \em average-bad \em if it forms an obstruction to average-radius list-recoverability.
\begin{definition}\label{def:bad}
For $X \subseteq \F^d$ of size $n$ and $\Lambda \subseteq \F^d$, we say that $(X,\Lambda)$ is $(L,d,\eps\ml)$-average-bad if:
\begin{itemize}
	\item[(a)] $X$ has full rank,
	\item[(b)] $|\Lambda| \geq L$,
	\item[(c)] $\dim(\spn(\Lambda)) \leq d$,
	\item[(d)] and
$ \sum_{x \in X} \mrk{\pl^{(\ell)}_x}(\Lambda) \ge \eps n. $
\end{itemize}
We say that a set $X$ is $(L,d,\eps\ml)$-average-bad if there exists a $\Lambda$ so that $(X,\Lambda)$ is $(L, d, \eps\ml)$-average-bad.  For an average-bad $X$, we will refer to such a $\Lambda$ as a \em witness \em for $X$'s badness.  
\end{definition}

As with all-badness and zero-error list-recovery, average-badness is related to average-radius list-recovery.
\begin{proposition}\label{prop:plrs}
A linear code $\cC$ is $(\mrk{\eps},\mrk{\ell} ,L)$-average radius list-\mrk{recoverable} if and only if 
for all $d \leq \dim(\cC)$,
there are no $X \subseteq \F^d$ so that
$X$ is $(L, d, \eps, \ell)$-average-bad and so that $\cols(X) \subseteq \cC$.
\end{proposition}
\begin{proof}

As discussed in Section~\ref{sec:notation}, for any full-rank set $X = \inset{x_1,\ldots,x_n} \subseteq \F^d$ of size $n$ with $\cols(X) \subseteq \cC$, and any $\Lambda \subseteq \F^d$ of size $L$, there is a set of $L$ codewords
\[ \inset{ \bX v \suchthat v \in \Lambda } = \inset{ ( \ip{x_1}{v}, \ldots, \ip{x_n}{v} ) \suchthat v \in \Lambda } \subseteq \cC, \]
where $\bX$ is the matrix with the elements of $X$ as rows.
Conversely, for any set of codewords we may come up with such a pair $(X,\Lambda)$. 
Thus, the condition of average-radius list-recoverability is the same as the condition that:
\begin{quote} \em For all $d \leq \dim(\cC)$, for all full-rank $X \subseteq \F^d$ with $\cols(X) \subseteq \cC$, for all $\Lambda \subseteq \F^d$ of size $L$ and for all $z\in\F^{\ell\times n}$, we have
\begin{equation}\label{eq:avgraddef}
\frac{1}{L} \sum_{v \in \Lambda} \dist( \bX v , z ) > (1 - \eps),
\end{equation}
where $\bX$ is the matrix with the elements of $X$ as rows.
\end{quote}
For a fixed $\Lambda$,
the $z \in \F^{\ell \times n}$ that maximizes the left-hand size is the one which has for all $j$,
\[ z[j][i] = \argtp_i(x_j, \Lambda). \]
Thus, \eqref{eq:avgraddef} is equivalent to the condition that:
\begin{quote}\em
For all $d \leq \dim(\cC)$, for all full-rank $X \subseteq \F^d$, and for all $\Lambda \subseteq \F^d$ of size $L$,
\begin{align*}
\eps & > 
\frac{1}{L} \sum_{v \in \Lambda} \frac{1}{n} \sum_{j \in [n]} \ind{ \ip{x_j}{v} = \argtp_i(x_j, \Lambda) \text{ for some $i \leq \ell$ } }  \\
& = \frac{1}{n} \sum_{j \in [n]} \frac{1}{L} \sum_{v \in \Lambda} \ind{ \ip{x_j}{v} = \argtp_i(x_j, \Lambda) \text{ for some $i \leq \ell$ } } \\
& = \frac{1}{n} \sum_{x \in X} \pl^{(\ell)}_x(\Lambda).
\end{align*}
\end{quote}
This is precisely the condition that for all $d \leq \dim(\cC)$, there are no $(L,d,\eps,\ell)$-average-bad sets $X \subset \F^d$ so that $\cols(X) \subseteq \cC$, and we have proved the Proposition.
\end{proof}

Given Proposition~\ref{prop:plrs}, we wish to show that for a random linear code $\cC$, 
with high probability, 
for all $d \leq Rn$ and 
for all sets $X$ that are $(L, d, \eps\ml)$-bad, $\cols(X) \not\subset \cC$.  
We will choose 
\[ R \leq \min \inset{ (\eps - \nmrk{\beta\cdot\centermu-\barmu})(1 - 5\zeta) - \eta, \arm{1 - H_{q/\ell}(1 - \eps - \eta) - \log_q(\ell)} - \xi} \]
as in the theorem statement.
Let 
\begin{equation}\label{eq:naughts}
	L_0 = \frac{\eta L}{1 - \eps + \eta}, \qquad \eps_0 = \eps - \eta, \qquad d_0 = \frac{2\mrk{\ell}}{\xi}.
\end{equation}
We begin by observing that if the dimension of $\Lambda$ is suitably large (larger than $d_0$), then we are done by independence: Lemma~\ref{lem:start} is the analog of Lemma~\ref{lem:base2} in Section~\ref{sec:easy}.

\begin{lemma}\label{lem:start} Suppose that $d \geq d_0$.  Let $\cC$ be a random linear code of rate $R$ over $\F_q$ that satisfies
\begin{equation}\label{eq:rate}
 R \leq \arm{1 - H_{q\mrk{/\ell}}( 1 - \eps+\eta )\mrk{- \log_q{\ell}}} - \xi.
\end{equation}
    Then with probability at least $1 - q^{-\xi n/2}$, if there is a set
 $X \subseteq \F^d$ that is $(L, d, \eps\ml)$-average-bad and has $\cols(X) \subseteq \cC$, then there is
    some set $X' \subseteq \F^{d_0}$ of size $n$,
    so that $\cols(X') \subseteq \cC$ and so that $X'$ is $(L_0, d_0, \eps-\eta\ml)$-average-bad.
\end{lemma}
\begin{proof}
%
    We begin by observing that the statement is true in the traditional list-recovery sense (without the ``average-case" modifier). 
    \begin{claim}\label{claim:vanilla}
        Suppose that the conditions of Theorem~\ref{thm:avgrad}, and suppose that $d > 2 \mrk{\ell}/\xi$.
        Then with probability at least $1 - q^{-\xi d n/2}$, there is no $X, \Lambda \subseteq \F^d$ of full rank, with $|X| = n$, so that
        \begin{itemize}
            \item[(a)] there is some $z \in \F^{\mrk{\ell \times}n}$ so that $\dist(z, \bX v) \leq 1 - \eps+\eta$ for all $v \in \Lambda$, where $\bX \in \F^{n \times d}$ is the matrix with the elements of $X$ as rows; and
            \item[(b)] $\cols(X) \subseteq \cC$.
        \end{itemize}
\end{claim}
\begin{proof}
    Let $\mathbf{G}$ denote the generator matrix of $\cC$, so $\mathbf{G}$ is random.  Then using the definition of $\cC$, we may re-write the probability as
    \[ \PR{ \exists X, \Lambda \suchthat \text{(a), (b)} } = \PR{ \exists \Lambda, \exists z \in \F^{\mrk{\ell \times} n} \text{ s.t. } \forall v \in \Lambda, \dist(z, \mathbf{G}v) \leq 1 - \eps+\eta }. \]
    Let $\Lambda' \subseteq \Lambda$ be a set of $d$ linearly independent vectors in $\Lambda$, so 
    \[ \PR{ \exists X, \Lambda \suchthat \text{ (a),(b) } } \leq \PR{ \exists \Lambda', \exists z \in \F^{\mrk{\ell \times} n} \text{ s.t. } \forall v \in \Lambda', \dist(z, \mathbf{G}v) \leq 1 - \eps +\eta}. \]
    Now, the vectors $\mathbf{G}v$ are independent uniformly random vectors for $v \in \Lambda'$.  For a fixed $z$ and such a vector $\mathbf{G}v$, we have
    \begin{align*} 
    \PR{ \dist(z, \mathbf{G}v ) \leq 1 - \eps+\eta } 
        &= \sum_{i=0}^{(1-\eps+\eta)n}\binom{n}{i}(\ell/q)^{n-i}(1-\ell/q)^i \\
        &= \inparen{ \frac{\ell}{q}}^n\cdot \sum_{i=0}^{(1-\eps+\eta)n}\binom{n}{i}(q/\ell-1)^i\\
     &= \inparen{ \frac{\ell}{q} }^n \cdot \vol_{q/\ell}^n(1 - \eps+\eta) \\
     &\leq q^{-n(\arm{1 - H_{q\mrk{/\ell}}(1 - \eps+\eta) \mrk{ - \log_q(\ell)} }) } 
\end{align*}
    Now, using independence and a union bound over all $q^{n\mrk{\ell}}$ choices for $z$ and all ${q^{Rn} \choose d} \leq q^{Rnd}$ choices of $\Lambda$, we can bound this by
    \begin{align*}
    \PR{ \exists \Lambda', \exists z \text{ s.t. } \forall v \in \Lambda', \dist( z, \mathbf{G} v ) \leq 1 -  \eps+\eta } 
                                                  &\leq q^{n\mrk{\ell}} \cdot { q^{Rn} \choose d } \cdot q^{-dn(\arm{1 - H_{q\mrk{/\ell}}(1-\eps+\eta)\mrk{-\log_q{\ell}}})} \\
                      &\leq q^{n (Rd - d(\arm{1 - H_{q\mrk{/\ell}}(1 -\eps+\eta)\mrk{-\log_q{\ell}}}) + \mrk{\ell}). }
\end{align*}
    Thus, as long as $R < \arm{1 - H_{q\mrk{/\ell}}(1 - \eps+\eta) \mrk{-\log_q{\ell}}} - \xi$, and $d > 2\mrk{\ell}/\xi$, this probability is at most $q^{-\xi d n/2 }.$
\end{proof}

Let $\mathcal{E}$ denote the event that the favorable case of Claim~\ref{claim:vanilla} occurs: that is, that there are no $X,\Lambda$ that satisfy (a) and (b) above.
We will show that if $\mathcal{E}$ occurs, then we are in business. 

\begin{claim}\label{claim:extendtoavg}
Suppose that $\mathcal{E}$ occurs.  Then for all $d > d_0$, if there exists an $(L,d,\eps,\ell)$-average-bad set $X$ so that $\cols(X)\subseteq\cC$, then there exists a $\Lambda'$ such that $(X,\Lambda')$ is  $(L_0, d, \eps_0, \ell)$-average-bad and $\dim(\Lambda')\le d_0$. 
\end{claim}
\begin{proof}  Pick $X \subseteq \F^d$ so that $X$ is full-rank and suppose that $\cols(X) \subseteq \cC$.
    Suppose that there is some set $\Lambda \subset \F^{d}$ of size $L$ and dimension $d$ so that
    \[ \sum_{x \in X} \mrk{\pl^{(\ell)}_x}(\Lambda) \geq \eps n, \]
    that is, so that $(X,\Lambda)$ is $(L,d,\eps\ml)$-average-bad.  As in Proposition~\ref{prop:plrs}, this is equivalent to the condition that there is some $z \in \F^{\mrk{\ell\times}n}$ so that
    \[\frac{1}{L} \sum_{v \in \Lambda} \dist( z, \bX v ) \leq 1 - \eps. \]
    Let $\Lambda_+ \subseteq \Lambda$ denote the set
    \[ \Lambda_+ = \inset{v \in \Lambda \suchthat \dist( \bX v, z ) \leq 1 - \eps+\eta}, \]
        and let $\Lambda_- = \Lambda \setminus \Lambda_+.$
        Then by our assumption that $\mathcal{E}$ holds, $\Lambda_+$ must have dimension less than $2\mrk{\ell}/\nmrk{\xi}$.
    \begin{subclaim} $|\Lambda_+| \geq L_0$.\end{subclaim}
        \begin{proof}
                We have
    \begin{align*}
        (1 - \eps) L &\geq \sum_{v \in \Lambda} \dist( z, \bX v ) \\
            &= \sum_{v \in \Lambda_+ } \dist( z, \bX v ) + \sum_{v \in \Lambda_- } \dist(z, \bX v) \\
        &\geq \sum_{v \in \Lambda_{-}} \dist( z, \bX v ) \\
         &\geq (1 - \eps+\eta) (L - |\Lambda_+| ) 
    \end{align*}
    and rearranging we have
    \begin{align*}
        \frac{ \eta }{ 1 - \eps+\eta } &\leq \frac{ |\Lambda_+|} {L},
\end{align*}
which implies that $|\Lambda_+| \geq \frac{\eta L }{1-\eps+\eta}$ as claimed.
        \end{proof}
        Now, $\Lambda_+$ has size at least $\eta L /(1 - \eps + \eta) = L_0$ and dimension at most $2\mrk{\ell}/\xi = d_0$.  Further, it  has by definition that
        \[ \frac{1}{L} \sum_{v \in \Lambda_+} \dist( \bX v, z ) \leq 1 - \eps+\eta \]
hence
        \begin{equation}\label{eq:stillbad}
         \sum_{x \in X} \mrk{\pl^{(\ell)}_x}(\Lambda_+) \geq n(\eps-\eta). 
    \end{equation}
    Recalling that $\eps_0 = \eps - \eta$, we see that $(X, \Lambda_+)$ is $(L_0, d, \eps_0, \ell)$-average-bad, and $\Lambda_+$ has dimension at most $d_0$.
\end{proof}

Finally, we complete the proof of Lemma~\ref{lem:start}.
Assume that $X$ is $(L,d,\eps,\ell)$-average-bad with $\cols(X)\subseteq \cC$. Then 
 if $\cE$ occurs (which by Claim~\ref{claim:vanilla} is with probability at least $1 - q^{-\xi n/2}$), then Claim~\ref{claim:extendtoavg} shows there exists $(X, \Lambda)$ that is $(L_0, d, \eps_0, \ell)$-average-bad with $\dim(\Lambda) \leq d_0$.  Then by  Lemma~\ref{lem:projection}, these can be projected down to obtain $(X', \Lambda')$ that are $(L_0, d_0, \eps_0, \ell)$-average-bad, as desired.
    \end{proof}

Lemma~\ref{lem:start} implies that, in our quest to show that there are no bad sets $X$, we may consider only the $X$ that are $(L_0, d, \eps'\ml)$-bad, for $d \leq d_0$ and $\eps' \leq \eps_0$.  To do this, we will follow an inductive argument.

With the definitions of $d_0, L_0$, and $\eps_0$ as in \eqref{eq:naughts}, let
\[j_{\max} = \frac{\log(d_0/\dmax)}{\log \inparen{ \frac{1}{1 - \zeta}} } \qquad \text{for}\qquad  \dmax := \frac{2}{\zeta} \inparen{ 1 + \log_q(2/\zeta)}. \]
We suppose without loss of generality (it will affect only the lower-order terms) that $j_{\max}$ is an integer.
Further define, for $j = 1, \ldots, j_{\max}$, 
\begin{equation}\label{eq:js}
	d_j = (1-\zeta)^j\cdot d_0 \qquad L_j = \alpha^{j} \inparen{ \eps_0(1 - \zeta) }^{2j/\gamma} L_0 \qquad \eps_j = (1 - \gamma)^j \eps_0,
 \end{equation}
where
\[ \gamma := \frac{\zeta}{j_{\max}}, \qquad \alpha := \min \inset{ \frac{1}{2 d_0 q^2}, \frac{ \zeta }{ e q } }. \] 
Notice that our definition of $\dmax$ is consistent between our two definitions.
\begin{proposition}\label{prop:base}
	Suppose that 
	\[ L_0 > q^{\dmax}\inparen{ \frac{1}{\alpha} }^{j_{\max}} \inparen{ \frac{1}{\eps_0(1 -\zeta)} }^{2j_{\max}^2/\zeta}. \]
	Then there are no $X \subseteq \F^{j_{\max}}$ that are $(L_{j_{\max}}, d_{j_{\max}}, \eps_{j_{\max}}\ml)$-average-bad.
\end{proposition}
\begin{proof}
We will show that there are no bad $X$ that are $(L_{j_{\max}}, d_{j_{\max}}, \eps_{j_{\max}}\ml)$-average-bad, simply because $L_{j_{\max}} > q^{\dmax}$; thus there are no sets $\Lambda$ that could possibly witness this badness.
More precisely, any set $\Lambda \subseteq \F^{\dmax}$ has
$|\Lambda| \leq q^{\dmax}$, and so if $\Lambda$ had size at least $L_{\max}$ we would have
\begin{align*}
q^{\dmax} \geq |\Lambda| &\geq L_{j_{\max}}\\
& \geq L_0 \cdot \alpha^{j_{\max}} \inparen{ \eps_0(1 - \zeta)}^{2j_{\max}/\gamma} \\
&=
 L_0 \cdot \alpha^{j_{\max}} \inparen{ \eps_0 (1 - \zeta)}^{2j_{\max}^2/\zeta} .
\end{align*}
Our choice of $L_0$ shows that this does not happen.
\end{proof}

With the base case out of the way, we proceed by induction, decrementing $j$.  We maintain the following inductive hypothesis.

\begin{quote}
\textbf{Hypothesis$(j)$}:
There is no
$(L_{j}, d_{j}, \eps_{j}\ml)$-average-bad set $X$ so that $\cols(X) \subseteq \cC$. 
\end{quote}
We've already established \textbf{Hypothesis}$(j_{\max})$, and the following lemma will allow us to bootstrap this up to \textbf{Hypothesis}$(0)$ with high probability.

\begin{lemma}\label{lem:mainavg} \nmrk{Let $\barmu\in\inset{0,\centermu}$ and $\beta=\arm{(q+1)}^{1/(d(1-\zeta))}$ if $\barmu=0$ and $\beta=2$ otherwise.} Let $j \in \inset{0,\ldots,j_{\max}}$.  There is a constant $C > 0$ so that the following hold.  Suppose that
\[ R \le ( \eps_j - \nmrk{\beta\cdot\centermu-\barmu})(1 - 4\zeta) \qquad \text{and} \qquad q \ge  \ell^{2/\zeta}. \] 
	Then with probability at least $1 - q^{ -C R \zeta \cdot n }$, at least one of the following holds:
\begin{itemize}
	\item	\textbf{Hypothesis}$(j)$ holds; or
	\item	 there is some $i > j$ so that \textbf{Hypothesis}$(i)$ does \em not \em hold.
\end{itemize}
\end{lemma}

Before we prove the lemma, we show how it suffices to complete the proof of Theorem~\ref{thm:avgrad}. The base case, Proposition~\ref{prop:base}, implies that 
\textbf{Hypothesis$(j_{\max})$} holds; there is no $(L_{j_{\max}}, d_{j_{\max}}, \eps_{j_{\max}}\ml)$-bad set with columns in $\cC$.

First, we establish that the hypotheses of Lemma~\ref{lem:mainavg} hold for all $j$.  
We recall that for $j = j_{\max}$, we have
\[ \eps_{j_{\max}} = (1 - \gamma)^{j_{\max}} \eps_0 \geq \eps_0(1 - \zeta) = (\eps - \eta)(1 - \zeta). \]
Thus, for all $j$, 
\[ \eps \leq \inparen{ \frac{\eps_j}{1 - \zeta} } + \eta.\]
Our first condition on $R$ from the statement of Theorem~\ref{thm:avgrad} implies that
\begin{align*}
R &< (\eps - \nmrk{\beta\cdot\centermu-\barmu})(1 - 5\zeta) - \eta \\
&\leq \inparen{ \eta + \frac{ \eps_j }{1 - \zeta} -  \nmrk{\beta\cdot\centermu-\barmu} }(1 - 5\zeta) - \eta\\
&\leq ( \eta + \eps_j -  \nmrk{\beta\cdot\centermu-\barmu}) ( 1 - 4\zeta ) - \eta \\
&\leq (\eps_j -  \nmrk{\beta\cdot\centermu-\barmu})(1 - 4\zeta),
\end{align*}
as needed by Lemma~\ref{lem:mainavg}.  The second condition on $R$ in Theorem~\ref{thm:avgrad} also immediately implies that the condition on $R$ in Lemma~\ref{lem:start} is satisfied.
The choice of $q$ similarly immediately satisfies the requirements of both Lemma~\ref{lem:mainavg} and Lemma~\ref{lem:start}. \nmrk{Further, the choice of $\beta$ in Theorem~\ref{thm:avgrad} implies the choice of $\beta$ in Lemma~\ref{lem:mainavg}-- in particular, note that $\arm{(q+1)}^{1/(d(1-\zeta))}\le \arm{(q+1)}^{1/(d_{\max}(1-\zeta))}\le \arm{(q+1)}^{\zeta/(2(1-\zeta))}$.}
Finally, from the assumption on $L$ and the definition \eqref{eq:naughts} of $L_0$, we have (for an appropriate choice of $C'$)
\begin{align*}
L_0 &\geq \exp_{q\nmrk{\cdot \ell/\xi}}\inparen{ \frac{ C' \log(\ell \zeta/\xi ) }{\zeta } } \cdot \inparen{\frac{1}{ \eps - \eta } }^{C' \log^2(\ell \zeta/\xi )/\zeta^3 }\\
&>
q^{\dmax} \inparen{ \frac{1}{\alpha} }^{\lg(d_0/\dmax)/\zeta} \inparen{ \frac{ 1}{\eps_0(1 - \zeta)} }^{\lg^2(d_0/\dmax)/\zeta^3} \\
&= q^{\dmax} \inparen{ \frac{1}{\alpha}}^{j_{\max}} \inparen{ \frac{ 1 }{ \eps_0 (1 - \zeta) } }^{ 2j_{\max}^2/\zeta },
\end{align*}
which means that Proposition~\ref{prop:base} holds.

Now, we proceed inductively.  For every $j$, Lemma~\ref{lem:mainavg} shows that
\begin{align*}
	\PR{ \mbox{\textbf{Hypothesis}($j$) does not hold} } &\leq q^{-CR\zeta n} + \PR{ \exists i > j, \mbox{\textbf{Hypothesis}($i$) does not hold} }\\
	&\leq q^{-CR\zeta n} + \sum_{i > j} \PR{ \mbox{\textbf{Hypothesis}($i$) does not hold } }.
\end{align*}
Using the fact (Proposition~\ref{prop:base}) that
\[ \PR{ \text{\textbf{Hypothesis}}(j_{\max}) } = 0, \]
we conclude that
\[ \PR{ \text{ \textbf{Hypothesis}$(0)$ does not hold} } \leq j_{\max}\cdot q^{-C R \zeta n } \leq q^{-C' R\zeta n}, \]
using the fact that $n$ is sufficiently large compared to $d_0$.  
Finally, Lemma~\ref{lem:start} implies that
\[ \PR{ \exists d \leq Rn, X \in \F^d, \cols(X) \subseteq \cC, X \text{ is } (L, d, \eps\ml)\text{-average-bad} } \leq q^{-\xi n/2 } + q^{- C' \zeta R n }. \] 
Applying Proposition~\ref{prop:plrs} completes the proof of Theorem~\ref{thm:avgrad}, modulo the proof of Lemma~\ref{lem:mainavg}.

\subsection{Proof of Lemma~\ref{lem:mainavg}}

We wish to show that, with high probability, either \textbf{Hypothesis}$(j)$ holds, or else \textbf{Hypothesis}$(i)$ does not hold for some $i > j$.
To do this, we study $\sigma_p(\Lambda)$, as in Definition~\ref{def:sigma}.  Lemma~\ref{lem:sigma} implies that if $\sigma_p(\Lambda)$ is large for some $p$, then we can find a large low-dimensional subset of $\Lambda$; this will mean that it's likely that \textbf{Hypothesis}$(i)$ does not hold for some $i > j$.  On the other hand, if $\sigma_p(\Lambda)$ is always small, then the probability of a bad set $X$ being contained in $\cC$ with witness $\Lambda$, is also small, and hence \textbf{Hypothesis}$(j)$ is likely to hold.
\paragraph{Notation:} For notational simplicity, we will drop the subscripts $j$ and $j+1$ throughout this section, and we will adopt the notation
$ d_j \to d, \eps_j \to \eps, L_j \to L.$  This overloads the original parameters $\eps$ and $L$ in the statement of Theorem~\ref{thm:avgrad}.
	Similarly, for this section only we denote $L_{j+1}, d_{j+1}, \eps_{j+1}$ by $L', d', \eps'$, so $L' = \alpha L (\eps_0(1 -\zeta))^{2/\gamma}$, $d' = d(1 -\zeta)$, and $\eps' = (1 - \gamma)\eps$.
\newline
	
Say that $\Lambda \subseteq \F^d$ is \textbf{good} if 
 \[\sigma_p(\Lambda) \le \begin{cases} 
\frac{ 1 + 1/q }{q^p} & p \leq (1 - \zeta)d \\
\frac{d}{q^{d(1 - 2\zeta)} \cdot \ell^p} & (1 - \zeta)d < p \leq d   \end{cases} \]
In order to count the number of bad $X$'s, we first bound the probability of the event that a random $X$ is bad with witness $\Lambda$, assuming that $\Lambda$ is good; later, we will handle the case for bad witnesses $\Lambda$.


\begin{lemma}\label{lem:prbadX}
\nmrk{
Let $\barmu\in \inset{ 0, \centermu}$. Define
\[\beta=\begin{cases}
\arm{(q+1)}^{\frac{1}{d(1-\zeta)}} & \text{ if } \barmu=0\\
2 & \text{ if } \barmu=\centermu
\end{cases}.
\]
}
Fix $\eps > \mu$ and set
	 $\Delta = \eps - \nmrk{\barmu}$.
	Suppose that $\Lambda \subseteq \F^d$ is good and has size $L$, and
	suppose that $X \subseteq \F^{d}$ of size $n$ is drawn uniformly at random with replacement.  Choose $\lambda>0$. 
Then
	\begin{align*}
		\mathbb{P}_X \inset{ (X,\Lambda) \text{ is } (L, d, \eps\ml)\text{-average-bad} } 
&\leq \inparen{ \exp\inparen{\lambda \nmrk{\inparen{ \beta\cdot\centermu-\eps+\barmu  } }} +\frac{ qd }{q^{d(1 - 2\zeta)}}\cdot  \exp(\lambda(1 - \Delta))}^n \\
&=: (p_0 + p_1)^n. 
\end{align*}
\end{lemma}
\begin{proof}
	We will bound the probability of a random $X$ being bad with certificate $\Lambda$ by examining the moments of $\pl_x^{(\ell)}(\Lambda)$, for a random $x$.
	For a set $X \subset \F^d$ of size $n$ whose elements are drawn independently, uniformly at random from $\F^d$, we have that for all $\lambda > 0$,
	\begin{align}
		\PR{ \sum_{x \in X} \mrk{\pl^{(\ell)}_x}(\Lambda)  \ge n\eps }
		&= \PR{ \sum_{x \in X} (\mrk{\pl^{(\ell)}_x(\Lambda)} - \nmrk{\barmu}) \ge n\Delta} \notag\\
		&= \PR{ \exp\inparen{ \lambda \sum_{x \in X} (\mrk{\pl^{(\ell)}_x}(\Lambda) - \nmrk{\barmu}) } \ge \exp(\lambda n\Delta ) }\notag \\
		&\leq \exp(-\lambda n \Delta ) \EE_X \exp \inparen{ \lambda  \sum_{x \in X} (\mrk{\pl^{(\ell)}_x}(\Lambda) - \nmrk{\barmu}) }\notag \\
		&= \exp(-\lambda n \Delta ) \prod_{x \in X} \inparen{\EE_x \exp( \lambda (\mrk{\pl^{(\ell)}_x}(\Lambda) - \nmrk{\barmu})) }\notag\\
		&= \exp(-\lambda n \Delta) \prod_{x \in X}\inparen{ \sum_{p=0}^\infty \EE_x \frac{\lambda^p (\mrk{\pl^{(\ell)}_x}(\Lambda) - \nmrk{\barmu})^p }{p!} } \label{eq:genbound}
	\end{align}
	using Markov's inequality and the series expansion of $\exp(x)$ about $x = 0$. 
Now, we will bound these moments by expressions involving $\sigma_p(\Lambda)$, and use our assumptions of goodness to bound those.  First, we handle $p \leq d(1 - \zeta)$.
\begin{claim}\label{claim:smallp}
Suppose that $\Lambda$ is good. Then 
for all $p < d(1 - \zeta)$ 
we have
    \[ \EE_x \inparen{ \pl{(\ell)}_x(\Lambda) - \nmrk{\centermu} }^p \leq \nmrk{\inparen{2\cdot\centermu}^p} \]
\end{claim}
\begin{proof}
    Suppose $x \in \F^d$ is chosen uniformly at random.

\begin{subclaim}\label{claim:itsallmu} 
For any $p > 0$,
    \[ \EE_X (\pl_x^{(\ell)}(\Lambda) - \barmu)^p \leq \max \inset{ \ell^p \EE(\pl_x(\Lambda) - \barmu/\ell)^p, \barmu^p }. \]
\end{subclaim}\begin{proof}
    To see this, consider $(\pl_x^{(\ell)}(\Lambda) - \barmu)$ (without the expectation).  Notice that if $\pl_x^{(\ell)}(\Lambda) \leq \barmu$, then this quantity is negative, and thus is clearly at most $\barmu$.  On the other hand, if $\pl_x^{(\ell)}(\Lambda) > \barmu$, then this quantity is positive, and since
\[ \pl_x^{(\ell)}(\Lambda) \leq \ell \cdot \pl_x(\Lambda), \]
we have
\[ (\pl_x^{(\ell)}(\Lambda) - \barmu)^p \leq \ell^p (\pl_x(\Lambda) - \barmu/\ell)^p. \]
Then the claim follows.
\end{proof}
Next, we will show that the first element inside the maximum is small:
\begin{subclaim}\label{claim:itsallmutoo}  Suppose that $\Lambda$ is good.  For all $p \leq (1 - \zeta)d$,
\[     \ell^p \EE(\pl_x(\Lambda) - \nmrk{1/q})^p \leq \nmrk{\inparen{2\cdot\centermu}^p}.\] 
\end{subclaim}
\begin{proof}
We compute
    \begin{align}
 \ell^p \EE\inparen{ \pl_x(\Lambda) - \onebyq}^p
        & = \ell^p \EE_x \inparen{\frac{1}{L}\cdot \sum_{v \in \Lambda} \inparen{\ind{ \ip{x}{v} = \mrk{\argmax(x,\Lambda)} } - \onebyq}}^p \notag\\
        &=
        \frac{\ell^p}{L^p} \EE_x \sum_{v_1,\ldots,v_p \in \Lambda} \prod_{i=1}^p \inparen{\ind{ \ip{v_i}{x} = \argmax(x,\Lambda)} - \onebyq} \notag\\
        &\leq \frac{\ell^p}{L^p} \EE_x  \sum_{\alpha \in \F} \sum_{v_1,\ldots,v_p \in \Lambda} \prod_{i=1}^p \inparen{\ind{ \ip{v_i}{x} = \alpha } - \onebyq} \notag\\
        &= \ell^p \sum_{\alpha \in \F}  \EE_{v_1,\ldots,v_p \in \Lambda} \EE_x \prod_{i=1}^p \inparen{\ind{ \ip{v_i}{x} = \alpha } - \onebyq}.  \label{eq:expand}
\end{align}
Expanding the product, we bound the above by
\begin{align*}
         \ell^p \sum_{\alpha \in \F}  &\EE_{v_1,\ldots,v_p \in \Lambda} \EE_x \prod_{i=1}^p \inparen{\ind{ \ip{v_i}{x} = \alpha } - \onebyq}\\
        &= \ell^p \sum_{\alpha \in \F}  \EE_{v_1,\ldots,v_p \in \Lambda} \EE_x \sum_{S \subseteq [p]} \prod_{i \in S}\inparen{\ind{ \ip{v_i}{x} = \alpha }\cdot\inparen{-\onebyq}^{p - |S|} } \\
        &= \ell^p \sum_{\alpha \in \F} \sum_{S \subseteq [p]} \inparen{-\onebyq}^{p - |S|} \EE_{v_1,\ldots, v_p \in \Lambda} \EE_x \prod_{i \in S} \ind{ \ip{v_i}{x} = \alpha} \\
        &\leq q \ell^p \sum_{S \subseteq [p]} \inparen{-\onebyq}^{p - |S|} \EE_{v_1,\ldots, v_p \in \Lambda} q^{-\dim(v_1,\ldots, v_p)} \\
        &= q \ell^p  \sum_{S \subseteq[p]} \inparen{-\onebyq}^{p - |S|} \sigma_{|S|}(\Lambda).
\end{align*}
Now, the fact that $\Lambda$ is good and that $|S| \leq p \leq (1 - \zeta)d$ implies that
\[ \frac{1}{q^{|S|}} \leq \sigma_{|S|}(\Lambda) \leq \frac{1}{q^{|S|}}\inparen{1 + \frac{1}{q}} \]
where the lower bound follows from the definition of $\sigma_p$. 
Thus, we can bound the above by
\begin{align*}
         q \ell^p  \sum_{S \subseteq[p]} \inparen{-\onebyq}^{p - |S|} \sigma_{|S|}(\Lambda)
        &\leq q\ell^p \inparen{ \inabs{\sum_{S \subseteq[p]} \inparen{-\onebyq}^{p- |S|} \inparen{ \frac{1}{q}}^{|S|} }
+ \inabs{ \sum_{S \subseteq [p] } \inparen{\onebyq}^{p - |S|} \inparen{\frac{1}{q}}^{|S| + 1} }} \\
        &= q \ell^p \inabs{ \frac{1}{q} - \onebyq}^p + \ell^p \inparen{ \frac{1}{q} + \onebyq }^p \\
	&\nmrk{=  \inparen{ 2\cdot \centermu }^p}.
\end{align*}
\end{proof}

Finally, Sub-Claims~\ref{claim:itsallmu} and \ref{claim:itsallmutoo} imply that
\[ \EE\inparen{ \pl^{(\ell)}_x(\Lambda) - \centermu }^p \leq \max\nmrk{\inset{  \inparen{2\cdot \centermu}^p, \inparen{\centermu}^p }  \le \inparen{2\cdot \centermu}^p}. \]
%
This proves Claim~\ref{claim:smallp}.
\end{proof}

\nmrk{
Next, we prove a bound on the uncentered moments:
\begin{claim} 
\label{clm:gen-p-centered-mom}
For any $p>0$, we have
\[\EE_x (\mrk{\pl^{(\ell)}_x}(\Lambda))^p \le q\ell^p \sigma_p(\Lambda).\]
\end{claim}
\begin{proof}
Similarly to the above, we may bound the uncentered moments of $\pl^{(\ell)}_x(\Lambda)$ in terms of $\sigma_p(\Lambda)$; more precisely, we have for all $p > 0$,
        \begin{align*}
                \EE_x (\mrk{\pl^{(\ell)}_x}(\Lambda))^p &\leq \ell^p \EE( \pl_x(\Lambda) )^p\\
                & = \ell^p \EE_x \inparen{\frac{1}{L}\cdot \sum_{v \in \Lambda} \ind{ \ip{x}{v} = \mrk{\argmax(x,\Lambda)} } }^p \\
                &=  
                \frac{\ell^p}{L^p} \EE_x \sum_{v_1,\ldots,v_p \in \Lambda} \prod_{i=1}^p \ind{ \ip{v_i}{x} = \argmax(x,\Lambda)}  \\
                &\leq \frac{\ell^p}{L^p} \EE_x \sum_{\alpha \in \F}  \sum_{v_1,\ldots,v_p \in \Lambda} \prod_{i=1}^p \ind{ \ip{v_i}{x} = \alpha }  \\ 
                &= \ell^p \sum_{\alpha \in \F}  \EE_{v_1,\ldots,v_p \in \Lambda} \EE_x \prod_{i=1}^p \ind{ \ip{v_i}{x} = \alpha }  \\
                &\leq \ell^p \sum_{\alpha \in \F}  \EE_{v_1,\ldots,v_p} q^{-\dim(v_1,\ldots,v_p)} \\
                &= q\ell^p \sigma_p(\Lambda).
        \end{align*}
\end{proof}
We now use the above to prove a sharper bound on the small $p$ uncentered moments:
\begin{claim} 
\label{clm:small-p-centered-mom} Suppose that $\Lambda$ is good.
For any $p\le d(1-\zeta)$, we have
\[\EE_x (\mrk{\pl^{(\ell)}_x}(\Lambda))^p \le (q+1)^{p/(d(1-\zeta))}\inparen{\centermu}^p.\]
\end{claim}
\begin{proof}
For notational convenience, define $r=d(1-\zeta)$. We will first relate the $p$th moment in terms of the $r$th moment bound. In particular,
{\allowdisplaybreaks
\begin{align*}
\EE_x (\mrk{\pl^{(\ell)}_x}(\Lambda))^p 
& \le \inparen{\EE_x (\mrk{\pl^{(\ell)}_x}(\Lambda))^r}^{p/r}\\
& \le \inparen{q\cdot\ell^r\cdot \sigma_r(\Lambda)}^{p/r}\\
& \le \inparen{\frac{(q+1)\ell^r}{q^r}}^{p/r}\\
& = (q+1)^{p/r}\cdot \inparen{\centermu}^p,
\end{align*}
}
as desired. In the above, the first inequality follows from Jensen's inequality (note that $z^{p/r}$ for $p<r$ is convex), the second inequality follows from Claim~\ref{clm:gen-p-centered-mom} while the final inequality follows from the fact that $\Lambda$ is good and $p\le r=d(1-\zeta)$.
\end{proof}
}

Next, we handle the remaining terms, where $p$ is large.
\begin{claim}\label{claim:bigp}
Let $\Lambda$ be good.
For all $p > (1 - \zeta)d$, 
\[ \EE_x \inparen{ \pl^{(\ell)}_x(\Lambda) - \centermu }^p \leq \max\inset{\inparen{\centermu}^p, d \cdot q^{1 - d(1 -2\zeta)}} . \]
\end{claim}
\begin{proof}
We first consider the uncentered moments $\EE_x \pl^{(\ell)}_x(\Lambda)^p$. \nmrk{By Claim~\ref{clm:gen-p-centered-mom}, we have 
\[\EE_x (\mrk{\pl^{(\ell)}_x}(\Lambda))^p \leq  q\ell^p \sigma_p(\Lambda).\]
}
In particular, for $d(1 - \zeta) < p \leq d$, we have
\begin{align*}
\EE_x (\pl^{(\ell)}_x(\Lambda))^p &\leq q \ell^p \cdot \frac{ d}{q^{d(1 - 2\zeta)}\cdot \ell^p } \\
&= d \cdot q^{1 - d(1 -2\zeta)},
\end{align*}
using the assumption that $\Lambda$ is good.
Since $\pl_x^{(\ell)}(\Lambda) \in (0,1)$, for any $p_2 > p_1$, we have
\[ \pl_x^{(\ell)}(\Lambda)^{p_1} \leq \pl_x^{(\ell)}(\Lambda)^{p_2}. \]
Thus, we may extend the above to all $p > d(1 - \zeta)$, and conclude that for all such $p$,
\begin{equation}\label{eq:bigp2}
 \EE_x(\pl^{(\ell)}_x(\Lambda))^p \leq d \cdot q^{1 - d(1 - 2\zeta)}. 
\end{equation}
We translate the centered moments above into the assumptions on $\sigma_p(\Lambda)$ (which are not centered) via the following observation:
\begin{subclaim}\label{claim:centering}
Let $Z$ be a positive random variable.  Then \nmrk{for any $\barmu>0$,}
\[ \EE(Z - \barmu)^p \leq \max\inset{ \barmu^p, \EE(Z^p) }. \]
\end{subclaim}
\begin{proof}
For any $A,B$, we have $(A - B)^p \leq |A^p - B^p|$.  Thus, 
\begin{align*}
\EE(Z - \barmu)^p &\leq \EE| Z^p - \barmu^p | 
\leq \max \inset{ \EE(Z^p), \barmu^p },
\end{align*}
where we have used in the final inequality that $Z$ and $\barmu$ are positive.
\end{proof}

Thus,
using Subclaim~\ref{claim:centering} with $Z = \pl^{(\ell)}_x(\Lambda)$ and \eqref{eq:bigp2},
we establish that for all $p > d(1 - \zeta)$ we have
\begin{align*}
 \EE_x\inparen{\pl^{(\ell)}_x(\Lambda) - \centermu}^p &\leq \max \inset{ d \cdot q^{1 - d(1 - 2\zeta)} , \inparen{\centermu}^p }.
\end{align*}
\end{proof}

Claims~\ref{claim:smallp},~\ref{clm:gen-p-centered-mom},~\ref{clm:small-p-centered-mom}, and~\ref{claim:bigp} show that 
\[ \EE_x \inparen{\mrk{\pl^{(\ell)}_x}(\Lambda) - \barmu }^p \leq 
\begin{cases} 
\nmrk{ \inparen{\beta\cdot\centermu}^p } & p \leq d(1 - \zeta) \\
 \nmrk{\inparen{\tau\cdot\centermu}^p} + d \cdot q^{1-d(1 - 2\zeta)}  & p > d(1 - \zeta),
\end{cases} \]
\nmrk{
where 
\[\tau=\begin{cases}
0 & \text{ if } \barmu=0\\
1 & \text{ if } \barmu=\centermu
\end{cases}.\]
Note that $\tau\le \beta$.
}
Plugging the above into \eqref{eq:genbound}, we obtain
	\begin{align*}
		&\PR{ \sum_{x \in X} \mrk{\pl^{(\ell)}_x}(\Lambda)  > n\eps }\\
		&\leq \exp(-\lambda n \Delta) \inparen{ \sum_{p=0}^{\infty} \inparen{ \frac{\nmrk{\inparen{\beta\cdot \centermu}^p} \lambda^p }{p!} } 
+ \frac{qd}{q^{d(1 - 2\zeta)}} \sum_{p>d(1 - \zeta)} \frac{\lambda^p}{p!} }^n \\
		&\leq \exp( -\lambda n \Delta) \inparen{ \exp\inparen{\nmrk{\beta\cdot \centermu\cdot}\lambda}
+ \frac{qd}{q^{d(1 - 2\zeta)}} \exp(\lambda) }^n\\
		&= \inparen{ \exp\inparen{ \lambda\inparen{ \nmrk{\beta\cdot\centermu - \eps+\barmu}}} +   \frac{ qd }{q^{d(1 - 2\zeta)}} \exp( \lambda(1 - \Delta) )}^n.
	\end{align*}
This establishes the Lemma.
\end{proof}

Now we may pick parameters to instantiate Lemma~\ref{lem:prbadX} to show that the good $\Lambda$'s do not pose a problem:
\begin{lemma}\label{lem:goodlambda}
There are constants $C,C'$ so that the following holds for sufficiently large $n$.
Choose $\eps$ so that $\nmrk{\barmu+\beta\cdot\centermu} < \eps < 1$, and choose any $\zeta \in (0,1/20)$ (where $\beta$ is as defined earlier). 
Suppose that
\[\nmrk{ R \le \inparen{\eps - \beta\cdot \centermu-\barmu}(1 - \nmrk{4}\zeta)}. \] 
Then the probability that there exists a good $\Lambda \subset \F^d$ of size $L$ and a set $X \subseteq \F^d$ so that $(X,\Lambda)$ is $({L},d,\eps\ml)$-bad and $\cols(X) \subseteq \cC$ is at most $q^{- C R\zeta n}$, for sufficiently large $n$.
\end{lemma}
\begin{proof}
Let $\Delta = \eps - \barmu$ as above.
We choose 
\[\lambda = \frac{R(1  + \zeta) d \log(q) }{\nmrk{\eps - \beta\cdot \centermu-\barmu}}\]
and apply Lemma~\ref{lem:prbadX}.
This implies
that for any good $\Lambda$,  the probability of a random  $X$ being $({L},d,\eps\ml)$-average-bad with good witness $\Lambda$ is at most
\[ p_{bad}^n := (p_0 + p_1 )^n \]
where the $p_i$'s are as defined in the conclusion of Lemma~\ref{lem:prbadX}.
Thus, the number of such bad $X$ is at most $q^{dn}\cdot p_{bad}^n$.
Now, for any fixed set $X$, the probability that $\cols(X)$ are contained in a random linear code $\cC \subseteq \F^n$ of rate $R$ is 
\[ \mathbb{P}_{\cC} \inset{ \cols(X) \subseteq \cC } \leq q^{-(1 - R)dn }. \]
Applying a union bound over all sets $X$ that are $(L,d,\eps\ml)$-average-bad for a fixed good witness $\Lambda$, and then a union bound over the at most ${q^d \choose L}$ good sets $\Lambda \subset \F^d$ of size $L$, 
\begin{align*}
	&\PR{ \exists\  (L,d,\eps\ml)\text{-average-bad } X \text{ with good witness } \Lambda \text{ and with } \cols(X) \subseteq \cC } \\
	&\qquad\leq {q^d \choose L} \cdot q^{dn} \cdot p_{bad}^n \cdot q^{-(1 - R)dn } \\
&\qquad \leq q^{dL} \cdot \inparen{q^{dR} \cdot p_{bad}}^n \\
&\qquad = q^{dL}\cdot \inparen{ q^{dR}p_0 + q^{dR}p_1 }^n. 
\end{align*} 
We will bound this by first bounding each of $q^{dR}p_0$ and $q^{dR}p_1$.  For the first term, we have
\begin{align}
q^{dR} p_0 &= \exp\inparen{-\lambda\inparen{ \nmrk{\eps -\beta\cdot\centermu-\barmu} }}\cdot q^{dR} \notag\\
&= q^{dR -d(R(1 + \zeta)) ) } \notag\\
&= q^{-dR\zeta}.\label{eq:final1}
\end{align}

For the second term,
\begin{align}
q^{dR} p_1 &=
\frac{qd}{q^{d(1 - 2\zeta)}} q^{dR} \exp( \lambda( 1 - \Delta ) )  \notag\\
&= qd \exp_q\inparen{ dR - d(1 - 2\zeta) + dR(1 + \zeta)\inparen{ \frac{ 1 - \Delta }{\nmrk{\Delta -\beta\cdot\centermu}  } } } \notag\\
&= qd \exp_q\inparen{ d \inparen{R\inparen{ 1 + \frac{(1 + \zeta)(1 - \Delta)}{\nmrk{\Delta -\beta\cdot\centermu}  } } -1 + 2\zeta  }} \notag\\
&\leq qd\exp_q \inparen{ d \inparen{ \nmrk{\inparen{\Delta-\beta\cdot\centermu}}(1 - 4\zeta)\inparen{1 + \frac{(1 + \zeta)(1 - \Delta)}{\nmrk{\Delta -\beta\cdot\centermu} }} - 1 + 2\zeta  }}\label{eq:userate}\\
&\nmrk{\le} qd \exp_q \inparen{ d \inparen{  (1 - 4\zeta)(1 + \zeta(1 - \Delta) ) -1  +2 \zeta }} \notag\\
&\leq qd\exp_q\inparen{d \inparen{  (1 - 4\zeta)(1  + \zeta) - 1 + 2\zeta }} \notag\\
&\leq qd \exp_q\inparen{ -d\zeta }\label{eq:final2}\\
\end{align}
where in \eqref{eq:userate} we have used the assumption that $R < (1 - 4\zeta)\inparen{\Delta-\beta\cdot\centermu}$.
Together, we have
\begin{align}
	&\PR{ \exists\  (L,d,\eps\ml)\text{-average-bad } X \text{ with good witness } \Lambda \text{ and with } \cols(X) \subseteq \cC } \notag\\
&\qquad \leq q^{dL} \cdot \inparen{ q^{dR}p_0 + q^{dR}p_1 }^n\notag\\
&\qquad \leq q^{dL} \inparen{ q^{-dR\zeta} + qd q^{-d\zeta }}^n \label{eq:usingprev}\\
&\qquad \leq q^{dL} \inparen{  q^{-C'dR\zeta }}^n \label{eq:usingzeta}\\
&\qquad \leq q^{-C R d \zeta n}, \label{eq:usingn}
\end{align}
for some constants $C',C$ and for sufficiently large $n$. 
Above, in \eqref{eq:usingprev} we have used \eqref{eq:final1} and \eqref{eq:final2}; in \eqref{eq:usingzeta}, we have used the fact that $d \geq d_{j_{\max}}$ (recall that for this proof, $d$ is equal to $d_j$, which is at least $d_{j_{\max}}$), as well as the definition of $\dmax$ and the fact that $\zeta$ is sufficiently small.
In the final line \eqref{eq:usingn}, we used the fact that since $p_{bad}$ does not depend on $n$, we may take $n$ sufficiently large (compared to $\eps, \ell, d$) to bound $q^{dL} p_{bad}^n \leq q^{-C R d \zeta n}$ for some constant $C$ which does not depend on $n$.
\end{proof}

Thus, as long as $\Lambda$ is good, we are in good shape: the probability that there is a bad $X$ with a good witness is very small.  We now turn our attention to the bad $\Lambda$'s.  We will show that if there is some bad $X$ with a bad witness, then with high probability, either \textbf{Hypothesis}$(i)$ is violated for some $i > j$, or else \textbf{Hypothesis}$(j)$ holds.
\begin{lemma}\label{lem:badlambdas}  There is a constant $C$ so that the following holds for sufficiently large $n$.
	Suppose that
	the hypotheses of Lemma~\ref{lem:mainavg} hold, and that $q \geq \ell^{2/\zeta}$.  Let $\cC \subset \F_q^n$ be a random linear code of rate $R$.
	Let $\Lambda \subseteq \F^d$ be bad. 
Then with probability at least $1 - q^{-C\zeta R n }$, at least one of the following is true:
\begin{itemize}
	\item[(A)] there is no $X \subseteq \F^d$ of size $n$ so that $(X,\Lambda)$ is $(L,d,\eps\ml)$-average-bad; or
	\item[(B)] \textbf{Hypothesis}$(i)$ fails to hold for $\cC$ for some $i \geq j$.
\end{itemize}
\end{lemma}
Before we prove the Lemma, we briefly recall parameters and give an overview of the proof.
We recall from \eqref{eq:js} that 
\[ L' = \alpha \inparen{  \eps_0 (1 - \zeta) }^{2/\gamma} L. \]
Now we will see the reason for the complicated expression for $L'$:  Lemma~\ref{lem:badlambdas} is performing the role of Claim~\ref{claim:easyversion} in Section~\ref{sec:easy}.  In that case, the argument was straightforward: if there were some all-bad set $X$ with a bad witness $\Lambda$ (violating (A)), then there was a large, low-dimensional subset of $\Lambda$; then this subset was also a witness for the all-badness of $X$, and we were done.

However, with average-badness, it is not so straightforward.  Indeed, it is not the case that if $(X, \Lambda)$ is average-bad, then $(X, \Lambda')$ is also average-bad for some subset $\Lambda' \subseteq \Lambda$: since $(X,\Lambda)$ is only bad on average, $\Lambda' \subseteq \Lambda$ might pick out the ``good" part.  In order to get around this, we will have to do a more complicated argument, with an extra layer of induction.

This inner induction will cause $L$ to shrink by a factor of $\alpha$ at each step for at most $\log(\alpha L/L')/\alpha$ steps; with each step, we will increase the average-badness of our list.  The choice of $L'$ is optimized so that when we do this (reaching a list size of $L'/\alpha$), we will have sufficient badness to reach a conclusion.
This will bring us to a list size that is at least $L'/\alpha$; the outer layer of induction will cause an additional $\alpha$ fraction of shrinkage, which will land the list size at $L'$, establishing Lemma~\ref{lem:mainavg}.

With the intuition out of the way, we proceed with the proof of Lemma~\ref{lem:badlambdas}.
\begin{proof}[Proof of Lemma~\ref{lem:badlambdas}]
The proof proceeds by induction with the following inductive hypothesis.
	\begin{quote}
\textbf{sub-Hypothesis}$(s)$:
Suppose that neither (A) nor (B) occur, and let $X \subseteq \F^d$ be a $(L,d,\eps,\ell)$-average-bad set with witness $\Lambda$ that violates (A).  Then one of the following two things occurs.
	\begin{itemize}
		\item[(a)] There is some $X' \subseteq \F^d$ of size $n$ so that $X'$ is 
\[ ( L'/\alpha, d, \eps\ml)\text{-average-bad} \]
with a \em good \em witness $\Lambda'$, and $\cols(X') \subseteq \cC$.
		\item[(b)] There is a subset $\tilde{\Lambda} \subseteq \Lambda$ with
			$|\tilde{\Lambda}| \geq (1 - \alpha)^sL$, 
	\[ \sum_{x \in X}\mrk{\pl^{(\ell)}_x}(\tilde{\Lambda}) \geq \eps n (1 + \gamma \alpha)^{s}. \]
	\end{itemize}
\end{quote}
	
	For the base case, $s = 0$, we choose $\tilde{\Lambda} = \Lambda$.  Then the second case (b) occurs: $\tilde{\Lambda}$ has size $L$ and the fact that $X$ is $(L,d,\eps\ml)$-average-bad implies that
	\[ \sum_{x \in X} \mrk{\pl^{(\ell)}_x}(\tilde{\Lambda})\geq \eps n. \]
	Thus, \textbf{sub-Hypothesis}$(0)$(b) holds.

	We will inductively peel sets $\Gamma$ off of $\tilde{\Lambda}$, using the following claims.
	\begin{claim} \label{claim:stillbad}
	If conclusion (B) does not occur, then 
	for all $X \subset \F^d$ of size $n$ with $\cols(X) \subseteq \cC$, and for all sets $\tilde{\Gamma}$ of dimension at most $d(1 -\zeta)$ and size at least $L'$,
	we have
	\[ \sum_{x \in X} \mrk{\pl^{(\ell)}_x}(\tilde{\Gamma}) \leq (1 - \gamma)\eps n. \]
\end{claim}
\begin{proof}
	If (B) does not hold, then \textbf{Hypothesis}$(j)$ does hold; so there are no $(L',d',\eps'\ml) = (L', d(1-\zeta), \eps(1- \gamma)\ml)$-average-bad sets $X'$ with $\cols(X') \subseteq \cC$.  
Thus, Lemma~\ref{lem:projection} gives the conclusion: there cannot be an $X, \tilde{\Gamma}$ with $\sum_{x \in X} \pl_x^{(\ell)}(\tilde{\Gamma})$ too small, or else we could project it down to find such an $X'$.

\end{proof}

	\begin{claim}\label{claim:stillstillbad} 
		Suppose that conclusion (B) does not occur, and
		suppose that $X$ is $(\tilde{L},d,\eps\ml)$-average-bad with witness $\tilde{\Lambda}$, for some $\tilde{L} \geq L'/\alpha$.  
		Let ${\Gamma} \subset \tilde{\Lambda}$ so that $|{\Gamma}| \geq |\tilde{\Lambda}|\alpha$, and $\dim({\Gamma}) \leq d(1 - \zeta).$
		Then
		\[ \sum_{x \in X} \mrk{\pl^{(\ell)}_x}(\tilde{\Lambda} \setminus {\Gamma}) \geq \eps n \inparen{1 + \gamma \alpha}.\]
	\end{claim}
	\begin{proof}
		Let $z \in \F^{\mrk{\ell\times}n}$ be the center of $\tilde{\Lambda}$, so that $z[j][i] =\argtp_i(x_j,\tilde{\Lambda})$.  Then by the definition of average-badness,
		\begin{align*}
			\eps n \tilde{L} 
&\leq \sum_{x \in X} \pl_x^{(\ell)}(\tilde{\Lambda})\cdot|\tilde{\Lambda}| \\
&\le |\Gamma| \sum_{x \in X} \pl_x^{(\ell)}(\Gamma) + ( |\tilde{\Lambda}| - |\Gamma| ) \sum_{x \in X} \pl_x^{(\ell)}(\tilde{\Lambda} \setminus \Gamma) \\
&\leq (1 - \gamma) \eps n |{\Gamma}| + (| \tilde{\Lambda}| - |\Gamma| ) \sum_{x \in X} \mrk{\pl^{(\ell)}_x}(\tilde{\Lambda} \setminus {\Gamma}),
		\end{align*}
		where we have used the fact that $|{\Gamma}| \geq \alpha |\tilde{\Lambda}| \geq L'$, and thus we can apply Claim~\ref{claim:stillbad}.
		Rearranging we have
		\[ \sum_{x \in X} \mrk{\pl^{(\ell)}_x} (\tilde{\Lambda} \setminus {\Gamma}) \geq \eps n \inparen{ 1+\frac{ \gamma (|{\Gamma}|/|\tilde{\Lambda}|) }{1 - |{\Gamma}|/|\tilde{\Lambda}| } } \geq \eps n \inparen{ 1 + \gamma \alpha}. \]
	\end{proof}

	Now, suppose that \textbf{sub-Hypothesis}$(s-1)$ holds, for 
	\[ 1 \leq s \leq \frac{ \log(\alpha L/L') }{\alpha}. \]
	If (a) holds for $s-1$, then it continues to hold for $s$, and so we establish \textbf{sub-Hypothesis}$(s)$.
	On the other hand, if (b) holds, there are two cases; either the set $\tilde{\Lambda}$ guaranteed by \textbf{sub-Hypothesis}$(s-1)$(b) is bad or it is good.
	\begin{itemize}
		\item Suppose $\tilde{\Lambda}$ is good.   By \textbf{sub-Hypothesis}$(s-1)$(b), $\tilde{\Lambda}$ is a witness for the $(|\tilde{\Lambda}|, d, \eps\ml)$-average-badness of $X$.  Notice that the dimension of $\tilde{\Lambda}$ is at most $d$ (since $\tilde{\Lambda} \subseteq \F^d$) and the size is at least 
\[ |\tilde{\Lambda}| \geq (1 - \alpha)^{s-1}L \geq \frac{L'}{\alpha} \]
using the bounds on $s$.
Thus, this establishes \textbf{sub-Hypothesis}$(s)$(a).
				
			\item Suppose $\tilde{\Lambda}$ is not good; thus, $(X,\tilde{\Lambda})$ is $(\tilde{L},d,\eps(1+\gamma\alpha)^{s-1}\ml)$-average-bad where $\tilde{L}\ge (1-\alpha)^{s-1}L \geq L'/\alpha$.
In this case, we may apply Lemma~\ref{lem:sigma} to find a set $\Gamma \subset \tilde{\Lambda}$ with size at least $|\tilde{\Lambda}|\alpha$ 
and with dimension at most $d(1 - \zeta)$. (Note that this is where we use the requirement that $q\ge \ell^{2/\zeta}$, which is required to apply Lemma~\ref{lem:sigma}.)  Notice that our hypotheses ensure that $q$ is large enough to use Lemma~\ref{lem:sigma}.  By Claim~\ref{claim:stillstillbad}, with this $\Gamma$, 
				\[ \sum_{x \in X} \mrk{\pl^{(\ell)}_x}(\tilde{\Lambda} \setminus \Gamma) \geq \eps \inparen{1 + \gamma \alpha}^{s}.\]
Further, since $|\tilde{\Lambda}| \geq (1 - \alpha)^{s-1} L $, we have 
\[|\tilde{\Lambda} \setminus \Gamma| = |\tilde{\Lambda}|(1 - \alpha) \geq (1 - \alpha)^s L. \]
		Updating $\tilde{\Lambda}$ to $\tilde{\Lambda} \setminus \Gamma$ establishes \textbf{sub-Hypothesis}$(s)$(b).
	\end{itemize}
Now, we carry out the induction up to	
	\[ s_{\max} = \log(\alpha L/L') / \alpha \]
	(at which point the inductive step stops working).
	We conclude that if neither (A) nor (B) hold, then
	either \textbf{sub-Hypothesis}$(s_{\max})$(a) or (b) holds.  We will show that this is unlikely, and conclude that at least one of (A) or (B) hold.

	First, by Lemma~\ref{lem:goodlambda}, the probability that (a) holds is at most $q^{-C \zeta R n}$ for some constant $C$.  Notice that we may apply Lemma~\ref{lem:goodlambda} to  $\tilde{\Lambda}$ since all of the relevant parameters ($d, \eps$) remain the same. 
Second, we claim that (b) cannot hold: 
using the definition
\[ L' = \alpha L \inparen{ \eps_0(1 - \zeta) }^{2/\gamma} \le \alpha L {\eps}^{2/\gamma}, \]
(recalling for this proof that $\eps = \eps_j \geq \eps_0(1 - \zeta)$), we see
\[ s_{\max} = \frac{ \log(\alpha L /L') }{\alpha} \geq \frac{ \log \inparen{ \frac{ \alpha L }{ \alpha L \eps^{2/\gamma} } } }{\alpha} = \frac{2}{\gamma \alpha} \log(1/\eps), \]
which in turn means that (b) reads as
	\[ \sum_x \mrk{\pl^{(\ell)}_x}(\tilde{\Lambda}) \geq n \eps ( 1 + \gamma \alpha )^{s_{\max}} 
		\geq n \eps \inparen{9/4 }^{ 2\log(1/\eps)  }
	 > n,\]
where above we have assumed that $\gamma < 1/2$ to assert that $(1 + \gamma \alpha)^{1/\gamma \alpha} \geq 9/4$.
	This is a contradiction since $\mrk{\pl_x^{(\ell)}}(\tilde{\Lambda}) \leq 1$, for all $x$, and there are at most $n$ things in the sum.

	This shows that for any bad $\Lambda \subseteq \F^d$, with probability at least $1 - q^{-C\zeta Rn}$, then at least one of (A) or (B) occurs.  This completes the proof of Lemma~\ref{lem:badlambdas}.
\end{proof}

Finally, Lemmas~\ref{lem:goodlambda} and \ref{lem:badlambdas} together imply Lemma~\ref{lem:mainavg}.  
To see this, we first fix a set $\Lambda$.  If $\Lambda$ is good, then by Lemma~\ref{lem:goodlambda}, with probability at least $1- q^{-CR\zeta n}$, \textbf{Hypothesis}$(j)$ holds.
If $\Lambda$ is bad, then by Lemma~\ref{lem:badlambdas}, with probability at least $1-q^{-C\zeta R n}$, either \textbf{Hypothesis}$(j)$ holds or else there is some $i \geq j$ so that \textbf{Hypothesis}$(i)$ does not hold.  Union bounding over all of the ${q^d \choose L}$ sets $\Lambda$ (and taking $n$ to be sufficiently large compared to $L, q, \ell, \zeta$) establishes Lemma~\ref{lem:mainavg}. 

\section{Conclusion}
We consider several instantiations of Question~\ref{q:bigq}---including list-decoding, average-radius list-decoding, and list-recovery, in a variety of parameter regimes---and provide a single argument that works for all of these.   Our argument can obtain improved results for list-decoding and list-recovery of low-rate codes, can establish list-recoverability of high-rate codes, and can also establish optimal average-radius list-decodability for constant-rate codes.  

Our work leaves several open questions. 
First, the obvious technical question is to reduce the quasipolynomial dependence on $\ell$ and the gap to capacity to optimal.   It seems as though this may be possible through a more careful treatment of our recursive argument.
A second and more philosophical question is about the nature of Question~\ref{q:bigq}.  Are there other instantiations of this question that our techniques can shed light on?  Trying to obtain applications in pseudorandomness is a natural candidate.

\section*{Acknowledgments}
We would like to thank Venkat Guruswami for helpful discussions.
\bibliographystyle{alpha}
\bibliography{refs}
\end{document}